\documentclass[PRO,english]{ipsj}
\usepackage{PROpresentation}
\PROheadtitle{2022-2-(6): Manuscript for presentation at IPSJ-SIGPRO, 28 07 2022.}

\usepackage{style} 

\usepackage{graphicx}
\usepackage{latexsym}

\def\Underline{\setbox0\hbox\bgroup\let\\\endUnderline}
\def\endUnderline{\vphantom{y}\egroup\smash{\underline{\box0}}\\}
\def\|{\verb|}

\received{2016}{3}{4}
\accepted{2016}{8}{1}

\usepackage[varg]{txfonts}
\makeatletter%
\input{ot1txtt.fd}
\makeatother%

\begin{document}

\title{Type checking data structures more complex than trees}

\affiliate{WSD}{Waseda University, Okubo, Shinjuku-ku, Tokyo, 169--8555, Japan}

\author{Jin Sano}{WSD}[sano@ueda.info.waseda.ac.jp]
\author{Naoki Yamamoto}{WSD}[yamamoto@ueda.info.waseda.ac.jp]
\author{Kazunori Ueda}{WSD}[ueda@ueda.info.waseda.ac.jp]

\begin{abstract}
	Graphs are a generalized concept that encompasses more complex data structures than trees,
	such as difference lists, doubly-linked lists, skip lists, and leaf-linked trees.
	Normally, these structures are handled with destructive assignments to heaps,
	which is opposed to a purely functional programming style and makes verification difficult.
	We propose a new
	purely functional language, \LamGT{}, that handles graphs as immutable,
	first-class data structures with a pattern matching mechanism
	based on Graph Transformation and developed a new type system, \(F_{GT}\), for the language.
	Our approach is in contrast with the analysis of pointer manipulation programs
	using separation logic, shape analysis, etc.\ in that
	(i) we do not consider destructive operations
	but pattern matchings over graphs provided by the new higher-level language that
	abstract pointers and heaps away and that
	(ii) we pursue what properties can be established automatically using a rather simple typing framework.
\end{abstract}

\begin{keyword}
	Functional programming,
	graph grammar,
	type system,
	program verification,
	heap analysis
\end{keyword}

\maketitle


\section{Introduction}\label{sec:introduction}

In this study, we propose a new functional language
that handle \emph{graphs} as a first-class data structure.
Graphs are a generalized concept
that encompasses more complex data structures than trees,
such as difference lists,
doubly-linked lists, skip lists \cite{skiplists},
and leaf-linked trees (\figref{fig:gallery}).
\begin{figure}[t]
	\centering
	\includegraphics[width=.9\hsize]{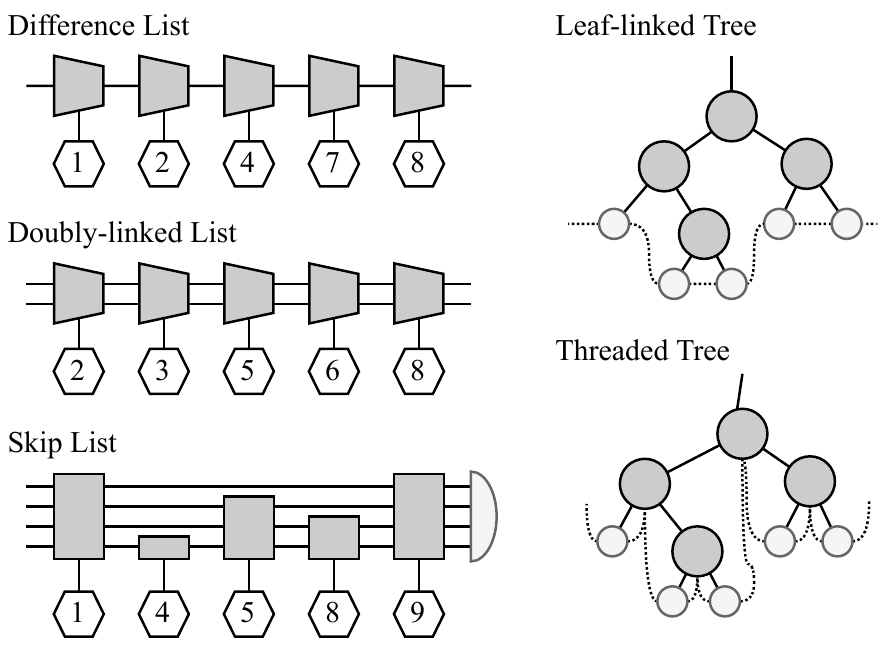}
	\caption{Examples of complex graph structures}\label{fig:gallery}
\end{figure}
However, graph structures cannot be handled succinctly in purely functional languages.
Although such structures can be handled with references,
this style implies imperative programming with destructive assignments,
which makes it hard to read and write programs and also makes verification more difficult.
In addition, classic type systems can only verify the types of the referenced data
and cannot verify the shape of the data structure.
Therefore, we aim to incorporate \emph{Graph Transformation} \cite{handbook_graph_grammar} to a functional language
and to develop a new type system for that.
Our approach is in contrast with the analysis of pointer manipulation
programs using
separation logic \cite{separation-logic}, shape analysis, etc.
in that (i) we consider graph structures formed by
higher-level languages that abstract pointers and heaps away and
guarantee low-level invariants such as the absence of dangling pointers and
that (ii) we pursue what properties can be established automatically
using a rather simple typing framework.

\subsection{HyperLMNtal: Hypergraph rewriting language}

Graph Transformation Systems (GTSs) are
computational models and
programming (or modeling) languages
based on graphs and their rewritings \cite{algebraic-gt,handbook_graph_grammar}.
Of various GTSs, HyperLMNtal \cite{hyperlmntal}
is a rewriting language that supports hypergraphs.
With hypergraphs, we can express structures more complex than trees,
e.g., difference lists, doubly-linked lists, skip lists, and leaf-linked trees.

HyperLMNtal allows us to handle these data structures declaratively with
rewrite rules that are activated by pattern matching.
Furthermore, GTS has cultivated a unique style of type checking frameworks
such as Structured Gamma \cite{structuredgamma}.
However,
GTSs are in general based on destructive rewriting
and do not support higher-order functions.
In contrast,
functional languages basically work with immutable data structures and
support higher-order functions,
making them highly modular.
This motivates us to study how we can
incorporate the data structure of HyperLMNtal into the \(\lambda\)-calculus.

\subsection{The \(\lambda_{GT}\) language}

We propose a new functional language,
\(\lambda_{GT}\),
that features graphs
as a first-class data structure.
The \(\lambda_{GT}\) language allows us to handle complex data structures
declaratively with a static type system.
Intuitively, the core language is
a call-by-value \(\lambda\)-calculus
that employs hypergraphs as values and supports pattern matching for them.

In order to formalize
hypergraphs in a syntax-directed manner, we employ the techniques
developed in a hypergraph rewriting language
HyperLMNtal \cite{hyperlmntal,sano2021}.
While various different formalisms have been proposed to handle the shapes
of graphs, including bisimulation (to handle ``equivalence'' of
cyclic structures) and morphism (in a category-theoretic approach), we believe
that our approach enables type checking relatively easily.
%
%
We also propose a new type-checking algorithm that automatically
performs this verification using structural induction.

\subsection{Contributions}

The main contributions of this paper are twofold.
\begin{enumerate}
	\item
	      We propose the formal syntax and semantics of
	      \(\lambda_{GT}\),
	      a pure functional language that handles data structures
	      beyond algebraic data types.
	\item
	      We propose a typing framework for the \(\lambda_{GT}\) language
	      and develop a new algorithm that can successfully handle
	      the manipulations of graphs,
	      which could not be handled in a previous study, Structured Gamma.
\end{enumerate}

\subsection{Structure of the Paper}

The rest of this paper is organized as follows.
Section~\ref{sec:hyperlmntal} introduces HyperLMNtal, a calculus model based on hypergraph transformation.
Section~\ref{sec:syntax-semantics} gives the syntax and the operational semantics of the proposing language \(\lambda_{GT}\).
Section~\ref{sec:type-synsem} introduces the new type system, \(F_{GT}\) proposed for \(\lambda_{GT}\).
Section~\ref{sec:fgt-ext} extends the system \(F_{GT}\) to cover powerful operations based on graph transformation.
Section~\ref{sec:autoverify} discusses the algorithm for the extended \(F_{GT}\).
Section~\ref{sec:related-work} describes related work.

\subsection{Syntactic conventions}

Throughout the paper, we use the following syntactic conventions.

For some syntactic entity \(E\),
\(\links{E}\) stands for a sequence
\(E_1,\dots,E_n\) for some \(n\ (\ge 0)\).
When we wish to mention the indices explicitly,
\(E_1,\dots,E_n\) will also be denoted as
\({\overrightarrow{E_i}}^i\).  The length of the sequence
\(\links{E}\) is denoted as \(|\links{E}|\).

For a set \(S\), the form \(S\!\{s\}\) stands for the set
\(S\) such that \(s\in S\) (or equivalently, \(S=S\cup\{s\}\)).

For some syntactic entities \(E\), \(p\) and \(q\),
a substitution \(E[q/p]\) stands for \(E\) with all the
(free) occurrences of \(p\) replaced by \(q\).
An explicit definition will be given if the substitution
should be capture-avoiding.
For substitutions of hyperlinks, we use a slightly different syntax
\(E\langle q/p\rangle\) for clarity.

In order to focus on novel and/or non-obvious aspects of the language,
constructs and properties that can be defined/derived in the same manner
as those of standard functional languages will be described rather briefly.

\section{HyperLMNtal}\label{sec:hyperlmntal}

HyperLMNtal is extended from LMNtal \cite{lmntal2009}.
LMNtal is a computational model and a programming language based on
hierarchical graph rewriting.
Flat LMNtal is a subset of LMNtal which does not allow a hierarchy of graphs.
Links in graphs that LMNtal handles are restricted to have at most two endpoints.
On the other hand, HyperLMNtal \cite{hyperlmntal} allows hyperlinks,
apart from normal links,
which can interconnect an arbitrary number of endpoints.
Flat HyperLMNtal is a subset of HyperLMNtal that disallow normal links and hierarchies of hypergraphs:
the data structure of Flat HyperLMNtal is formed only by
hyperlinks and nodes.

In the previous study, we have given syntax-directed semantics for Flat HyperLMNtal \cite{sano2021,sano-ba}.
As far as we have surveyed,
Flat HyperLMNtal is the only computational model that has syntax-directed semantics
which handles hypergraph matching and rewriting.
Since the \(\lambda\)-calculus and many other
computational models derived from the \(\lambda\)-calculus are
are defined as Structural Operational Semantics (SOS) \cite{sos},
it would be smoother to incorporate Flat HyperLMNtal than other
graph transformation formalisms based on algebraic approaches \cite{handbook_graph_grammar}.

The following subsections are based on Flat HyperLMNtal,
except that hypergraphs and rewrite rules are separated from each other
for the sake of formulation.
Hereinafter we simply refer to this language as \emph{HyperLMNtal},
hyperlinks as \emph{links},
and hypergraphs as \emph{graphs}.

\subsection{Syntax of graphs and rewrite rules}\label{sec:flat-hyperlmntal-syntax}

HyperLMNtal is composed of two syntactic categories.
\begin{itemize}
	\item
	      \(X\)
	      denotes a \emph{Link Name}.

	\item
	      \(p\)
	      denotes an \emph{Atom Name}.

\end{itemize}
The only preserved atom name is
\(\bowtie\), where an atom \(X \bowtie Y\), called a \emph{fusion},
fuses the link \(X\) and the link \(Y\) into a single link.

The syntax of HyperLMNtal is given in \figref{table:hyperlmntal-syntax}.
We abbreviate \(\nu X_1. \dots \nu X_n. G\) to \(\nu X_1 \dots X_n. G\),
which can be denoted as \(\nu \links{X}. G\).
The pair of the name \(p\) and the arity \(n = \norm{\links{X}}\) of an atom \(p(\links{X})\)
is referred to as the \emph{functor}\footnote{%
	Synonym of function symbol and function object; not to be confused with functors in category theory.
}
of the atom and is written as \(p/n\).

The set of free link names in hypergraph \(G\) is denoted as \(\mathit{fn}{(G)}\),
which is defined inductively in \figref{table:free-names}.

\begin{figure}[t]
	\normalsize
	\hrulefill{}
	\begin{center}
		\begin{tabular}{rcll}%
			\multicolumn{3}{@{}l}{Graph}                             \\
			\(G\) & $::=$   & \(\zero\)         & Null               \\
			      & $|$     & \(p (\links{X})\) & Atom               \\
			      & \(|\)   & \((G, G)\)        & Molecule           \\
			      & \(|\)   & \(\nu X.G\)       & Hyperlink creation \\
			\\
			\multicolumn{3}{@{}l}{Rewrite Rule}                      \\
			\(r\) & \(::=\) & \(G \means G\)    & Rule               \\
		\end{tabular}
	\end{center}
	\hrulefill{}
	\caption{Syntax of HyperLMNtal}\label{table:hyperlmntal-syntax}
\end{figure}

\begin{figure}[t]
	\normalsize
	\hrulefill{}
	\[
		\begin{aligned}
			\mathit{fn}(\zero)        & = \emptyset                              \\
			\mathit{fn}(p(\links{X})) & = \{\links{X}\}                          \\
			\mathit{fn}((G_1, G_2))   & = \mathit{fn}(G_1) \cup \mathit{fn}(G_2) \\
			\mathit{fn}(\nu X.G)      & = \mathit{fn}(G) \setminus \{X\}         \\
		\end{aligned}
	\]
	\hrulefill{}
	\caption{The set of free link names}\label{table:free-names}
\end{figure}

\begin{definition}[Abbreviation]\label{def:abbreviation}

	We introduce the following abbreviation schemes:
	\begin{enumerate}
		\item
		      A nullary atom \(p ()\) can be simply written as \(p\).


		\item
		      Term Notation:
		      \(\nu Y. (p (\links{X}, Y, \links{Z}), q (\links{W}, Y))\)
		      where
		      \(Y \notin \{\links{X}, \links{Z}, \links{W}\}\)
		      can be written as
		      \(p (\links{X}, q (\links{W}), \links{Z})\).

	\end{enumerate}
\end{definition}

Rules have the form \(G \means G\).
The two \(G\)s are called the left-hand side (LHS) and
the right-hand side (RHS), respectively.

\begin{definition}[Syntactic condition on rules]
	A rule \(G_1 \means G_2\) should satisfy
	\(\mathit{fn}(G_1) \supseteq \mathit{fn}(G_2)\).
\end{definition}

The condition indicates that we must denote a new hyperlink
in the scope of a \(\nu\) (new) on the RHS of a rule.

\subsection{Structural Congruence}\label{sec:congruence}

The semantics of Flat HyperLMNtal comes with two major ingredients,
structural congruence \(\equiv\)
and reduction relation \(\rightsquigarrow\) on graphs.
Structural congruence defines what graphs (represented in the
syntax of \figref{table:hyperlmntal-syntax}) are essentially the same.
This subsection defines structural congruence.

\begin{definition}[Link Substitution]

	\(G\angled{ Y_1, \dots, Y_n / X_1, \dots, X_n }\)
	is a \emph{link substitution} that replaces all free occurrences of
	\(X_i\) with \(Y_i\)
	as defined in \figref{table:hyperlink-substitution}.
	Here, the \(X_1, \dots, X_n\) should be mutually distinct.
	Note that, if a free occurrence of \(X_i\) occurs at a location where \(Y_i\) would not be free,
	\(\alpha\)-conversion may be required.

	\begin{figure}[t]
		\normalsize
		\hrulefill{}
		\begin{center}

			\begin{tabular}{@{}r@{\hspace{0.5em}}c@{\hspace{0.5em}}l@{}}
				\(\zero\sigma\)                         & \(=\) & \(\zero\)
				\vspace{0.8em}                                                                         \\

				\(p(\links{X})\sigma\)                  & \(=\) & \(p(X_1\sigma, \ldots, X_n\sigma) \) \\
				                                        &
				\multicolumn{2}{l}{%
					where \(
					X\angled{\links{Z}/\links{Y}} =
					\left\{
					\begin{array}{ll}
						Z_i & \mbox{if } X = Y_i                \\
						X   & \mbox{if } X \notin \{\links{Y}\}
					\end{array}
					\right.
					\)
					\vspace{0.8em} }
				\\

				\((G_1, G_2)\sigma\)                    & \(=\) & \((G_1 \sigma, G_2 \sigma)\)
				\vspace{0.8em}                                                                         \\
				\(
				(\nu X.G)\angled{\links{Z}/\links{Y}}\) & \(=\) &                                      \\
				\multicolumn{3}{l}{%
					\(\left\{
					\begin{array}{ll}
						\nu X.G\angled{\links{Z'}/\links{Y'}} & \mbox{if } X = Y_i\ \land                                      \\
						                                      & \links{Z'} = Z_1, \dots, Z_{i - 1}, Z_{i + 1}, \dots, Z_n      \\
						                                      & \links{Y'} = Y_1, \dots, Y_{i - 1}, Y_{i + 1}, \dots, Y_n
						\\[2mm]
						\nu X.G\angled{\links{Z}/\links{Y}}   & \mbox{if } X \notin \{\links{Y}\} \land X \notin \{\links{Z}\}
						\\[2mm]
						\nu W.(G\angled{W/X})\angled{\links{Z}/\links{Y}}
						                                      & \mbox{if } X \notin \{\links{Y}\} \land X \in \{\links{Z}\}    \\
						                                      & \land W \notin \mathit{fn}(G) \land W \notin \{\links{Z}\}
					\end{array}
					\right.
					\)
				}
			\end{tabular}
		\end{center}
		\hrulefill{}
		\caption{Hyperlink Substitution}\label{table:hyperlink-substitution}
	\end{figure}
\end{definition}

\begin{definition}[Structural Congruence]
	We define the relation \(\equiv\) on graphs as the minimal equivalence relation
	satisfying the rules shown in \figref{table:hyperlmntal-equiv}.
	Two graphs related by \(\equiv\) are essentially the same and are convertible
	to each other in zero steps.
	(E1), (E2) and (E3) are the characterization of molecules as multisets.
	(E4) and (E5) are structural rules that make \(\equiv\) a congruence.
	(E6) and (E7) are concerned with fusions.
	(E7) says that a closed fusion is equivalent to \(\zero\).
	(E6) is an absorption law of \(\bowtie\),
	which says that a fusion can be absorbed by connecting hyperlinks.
	Because of the symmetry of \(\bowtie\),
	(E6) says that an atom can emit a fusion as well.
	(E8), (E9) and (E10) are concerned with hyperlink creations.

	\begin{figure}[t]
		\normalsize{}
		\hrulefill{}
		\begin{center}
			\begin{tabular}{lrcl}
				(E1)  & \((\mathbf{0}, G)\)                                                          & \(\equiv\)      & \(G\)                            \\[1mm]
				(E2)  & \((G_1, G_2)\)                                                               & \(\equiv\)      & \((G_2, G_1)\)                   \\[1mm]
				(E3)  & \((G_1, (G_2, G_3))\)                                                        & \(\equiv\)      & \(((G_1, G_2), G_3)\)            \\[1mm]
				(E4)  & \(G_1 \equiv G_2\)                                                           & \(\Rightarrow\) & \((G_1, G_3) \equiv (G_2, G_3)\) \\[1mm]
				(E5)  & \(G_1 \equiv G_2\)                                                           & \(\Rightarrow\) & \(\nu X.G_1 \equiv \nu X.G_2\)   \\[1mm]
				(E6)  & \(\nu X.(X \bowtie Y, G)\)                                                   & \(\equiv\)      & \(\nu X.G\angled{Y / X}\)        \\
				      & \multicolumn{3}{l}{where \(X \in \mathit{fn}(G) \lor Y \in \mathit{fn}(G)\)}                                                      \\[1mm]
				(E7)  & \(\nu X.\nu Y.X \bowtie Y\)                                                  & \(\equiv\)      & \(\zero\)                        \\[1mm]
				(E8)  & \(\nu X.\zero\)                                                              & \(\equiv\)      & \(\zero\)                        \\[1mm]
				(E9)  & \(\nu X.\nu Y.G\)                                                            & \(\equiv\)      & \(\nu Y.\nu X.G\)                \\[1mm]
				(E10) & \(\nu X.(G_1, G_2)\)                                                         & \(\equiv\)      & \((\nu X.G_1, G_2)\)             \\
				      & \multicolumn{3}{l}{where \(X \notin \mathit{fn}(G_2)\)}                                                                           \\
			\end{tabular}
		\end{center}
		\hrulefill{}
		\caption{Structural congruence on HyperLMNtal graphs}\label{table:hyperlmntal-equiv}
	\end{figure}

\end{definition}

We give two important theorems showing that the symmetry of
\(\bowtie\) and \(\alpha\)-conversion can be derived from the rules of
\figref{table:hyperlmntal-equiv}.

\begin{theorem}[Symmetry of \(\bowtie\)]\label{th:symmetry-of-bowtie}

	\[X \bowtie Y \equiv Y \bowtie X\]

\end{theorem}
\begin{proof}
	See Chapter 3 of \Cite{sano-ba}.
\end{proof}

Thus, (E6) can be used also when we have a local link on
the right-hand side of \(\bowtie\).

\begin{theorem}[\(\alpha\)-conversion of hyperlinks]\label{th:alpha-equiv}

	Bound link names are \(\alpha\)-convertible in HyperLMNtal, i.e.,
	\[\nu X.G \equiv \nu Y.G\angled{Y/X} \text{ where } Y \notin \mathit{fn} (G)\]
\end{theorem}
\begin{proof}
	See Chapter 3 of \Cite{sano-ba}.
\end{proof}

{


	It is a subject for future work to
	elucidate the relationship between the structural congruence rules and
	graph isomorphism, including
	the completeness and the soundness of the structural congruence rules.
	However, these properties are irrelevant to the validity of the semantics of HyperLMNtal and \LamGT{},
	and the verification upon them.
	The graphs handled in HyperLMNtal and \LamGT{} are
	the graphs of HyperLMNtal defined inductively from the beginning,
	not the graphs
	in ordinary algebraic graph transformation formalisms \cite{algebraic-gt}.
	Therefore, there is no need for the structural congruence rules to correspond to
	graph isomorphism, and relating them would require another new formulation of HyperLMNtal graphs in the style of standard graph theory, which is beyond the scope of the present work.
	We have run examples to confirm that the equivalence using
	structural congruence is practical on HyperLMNtal and \LamGT{}.
}

\subsection{Reduction Relation}\label{sec:flat-hyperlmntal-op-sem}

We give the reduction relation of Flat HyperLMNtal that defined
the small-step semantics of the language.  Note, however, that
\(\lambda_{GT}\) described in the next section has its own operational
semantics without incorporating the reduction relation described here.
We nevertheless introduce the reduction relation of Flat HyperLMNtal here
because it serves as the basis of the graph types of
\(\lambda_{GT}\) described in Section~\ref{sec:type-synsem}.

\begin{definition}[Reduction relation]

	For a set \{P\} of rewrite rules,
	the reduction relation \(\rightsquigarrow_P\) on graphs is defined
	as the minimal relation
	satisfying the rules in \figref{table:hyperlmntal-trans}.

	\begin{figure}[h]
		\normalsize
		\hrulefill{}
		\vspace{1mm}
		\begin{center}
			\begin{tabular}{lc}
				(R1) & \(\dfrac{G_1 \rightsquigarrow_{P} G_2}{(G_1, G_3) \rightsquigarrow_{P}  (G_2, G_3)} \)
				\vspace{1em}                                                                                                                                                     \\
				(R2) & \(\dfrac{G_1 \rightsquigarrow_{P} G_2}{\nu X.G_1 \rightsquigarrow_{P}  \nu X.G_2} \)
				\vspace{1em}                                                                                                                                                     \\
				(R3) & \(\dfrac{\hspace{1em} G_1 \equiv G_2 \hspace{2em} G_2 \rightsquigarrow_{P} G_3 \hspace{2em} G_3 \equiv G_4 \hspace{1em}}{G_1 \rightsquigarrow_{P} G_4} \)
				\vspace{1em}                                                                                                                                                     \\
				(R4) & \(\dfrac{(G_1 \means G_2) \in P}{G_1 \rightsquigarrow_{P} G_2} \)                                                                                         \\
			\end{tabular}
		\end{center}
		\hrulefill{}
		\caption{Reduction relation on HyperLMNtal graphs}\label{table:hyperlmntal-trans}
	\end{figure}
\end{definition}

\section{Syntax and semantics of \(\lambda_{GT}\)}\label{sec:syntax-semantics}


This section describes the syntax and the semantics of \(\lambda_{GT}\),
which is a small, call-by-value functional language
that employs hypergraphs as values and supports pattern matching for them.
The main design issue is how to represent and manipulate hypergraphs
in the setting of a functional language and how to let hypergraphs and
abstractions co-exist in a unified framework.

\subsection{Syntax of \(\lambda_{GT}\)}\label{sub:nu-lang-syntax}

The \(\lambda_{GT}\) language is composed of the following syntactic categories.

\begin{itemize}
	\item
	      \(X\)
	      denotes a \emph{Link Name}.
	\item
	      \(C\)
	      denotes a \emph{Constructor Name}.

	\item
	      \(x\)
	      denotes a \emph{Graph Context Name}.
\end{itemize}

The syntax of the language is given in \figref{table:lgt-syntax}.
\(T\) is a \emph{template} of a graph.
It extends graphs in HyperLMNtal defined in \figref{table:hyperlmntal-syntax}
with \emph{graph contexts}.
A graph context \(x [\links{X}]\),
where \(\links{X}\) is a sequence of different links,
is a wildcard in pattern matching
corresponding to a variable in functional languages,
It matches any graph with free links \(\links{X}\).
Free links of a graph could be thought of as named parameters
(or `access points') of the graph.
\(C (\links{X})\) is a constructor atom.
Intuitively, it is a node of a data structure with links \(\links{X}\).
We allow
	{
		\(\lambda\)-abstractions
	}
as the names of atoms in graph templates \(T\)
(and its subclass \(G\) to be defined shortly).
The \(\lambda\)-abstraction atoms have the form
\((\lambda\, x [\links{X}]. e) (\links{Y})\).
Intuitively, the atom takes a graph with free links \(\links{X}\),
binds it to the graph
context
\(x [\links{X}]\) and
returns the value
(defined in \figref{table:lgt-value})
obtained by evaluating the expression \(e\)
with the bound graph context.
Notice that the \(\lambda\)-abstraction \((\lambda\, x [\links{X}].e) \) is \emph{just the name of an atom}:
\(\lambda\)-abstraction atoms can be incorporated into data structures
just like atoms with constructor names.
This is how \(\lambda_{GT}\) supports first-class functions in a
graph setting.
The free link(s) \(\links{Y}\) of the atom can be used to connect the atom to other structures such as lists
to form a graph structure containing a first-class function.
The links \(\links{X}\) and the links appearing in the graphs of the body expression \(e\)
are \emph{not} the free links of the atom.

\((\caseof{e_1}{T}{e_2}{e_3})\)
evaluates \(e_1\), checks whether this matches the graph template \(T\),
and reduces to \(e_2\) or \(e_3\).
The details are described in Sections~\ref{sub:graph-subsitution}--%
\ref{sec:reduction}.
The case expression covers just two cases in pattern matching,
but we can nest the expression to handle more cases.
\((e_1\; e_2)\) is an application.

Note that some graph rewriting languages including Interaction Nets \cite{interaction-net-encoding-lamnda}
and HyperLMNtal have encodings of the \(\lambda\)-calculus \cite{Machie-IN,
	hyperlmntal-lambda}
in which both abstractions and applications are \emph{encoded} using
explicit graph nodes and (hyper)~links.
In contrast, \(\lambda_{GT}\) features abstractions and applications at
the language level so as to retain the standard framework of functional
languages.

\(G\) stands for a \emph{value} of the language \(\lambda_{GT}\),
which is \(T\) not containing graph contexts.
%
Henceforth, we may call both \(G\) and \(T\) a \emph{graph} when
the distinction is not important.

\begin{figure}[t]
	\normalsize{}%
	\hrulefill{}%
	\vspace{-4mm}
	\begin{center}%
		\begin{tabular}{rcll}
			\multicolumn{3}{@{}l}{Graph Template}                               \\
			\(T\) & \(::=\) & \(\zero\)                    & Null               \\
			      & \(|\)   & \(x [\links{X}]\)            & Graph context      \\
			      & \(|\)   & \(v\, (\links{X})\)          & Atom               \\
			      & \(|\)   & \((T, T)\)                   & Molecule           \\
			      & \(|\)   & \(\nu X.T\)                  & Hyperlink creation \\
			\\
			\multicolumn{3}{@{}l}{Atom Name}                                    \\
			\(v\) & \(::=\) & \(C\)                        & Constructor name   \\
			      & \(|\)   & \(\lambda\, x[\links{X}].e\) & Abstraction        \\
			      & \(|\)   & \(\bowtie\)                  & Fusion             \\
			\\
			\multicolumn{3}{@{}l}{Expression}                                   \\
			\(e\) & \(::=\) & \(T\)                        & Graph              \\
			      & \(|\)   & \(\caseof{e}{T}{e}{e}\)      & Case               \\
			      & \(|\)   & \((e\; e)\)                  & Application        \\
		\end{tabular}
	\end{center}
	\hrulefill{}
	\caption{Syntax of \(\lambda_{GT}\)}\label{table:lgt-syntax}
\end{figure}

\begin{figure}[t]
	\normalsize
	\hrulefill{}
	\begin{center}
		\begin{tabular}{rcll}
			\multicolumn{3}{@{}l}{Value}                               \\
			\(G\) & \(::=\) & \(\zero\)           & Null               \\
			      & \(|\)   & \(v\, (\links{X})\) & Atom               \\
			      & \(|\)   & \((G, G)\)          & Molecule           \\
			      & \(|\)   & \(\nu X.G\)         & Hyperlink creation \\
		\end{tabular}
	\end{center}
	\hrulefill{}
	\caption{Value of \(\lambda_{GT}\)}\label{table:lgt-value}
\end{figure}

\begin{definition}[Syntactic condition on expressions]\label{def:no-abstraction}

	A \(\lambda\)-abstraction atom is not allowed to appear in the pattern \(T\) of the
	case expression \(\caseof{e_1}{T}{e_2}{e_3}\).

\end{definition}

\begin{definition}[Abbreviation rules for graph contexts]

	We introduce the following abbreviation schemes to graph contexts
	as well as we have done to atoms.
	\begin{enumerate}
		\item
		      The parentheses of nullary graph contexts can be abbreviated.
		      For example, \(x ()\) can be abbreviated as \(x\).

		\item
		      Term Notation: \(\nu X. (v\, (\dots, X, \dots), x[\dots, X])\)
		      can be abbreviated as
		      \(v\, (\dots, x[\dots], \dots)\).
		      The same can be done for embedding atoms (or graph contexts)
		      in the argument of a graph context (or atoms),
		      respectively.
	\end{enumerate}

\end{definition}

\begin{definition}[Free functors of an expression]\label{def:free-functor}

	We define free functors of an expression \(e\), \(\mathit{ff}(e)\),
	in \figref{table:free-functor}.
	Free functors are not to be confused with free link names.

	\begin{figure}[t]
		\normalsize
		\hrulefill{}
		\begin{center}
			\begin{tabular}{rcl}
				\multicolumn{3}{l}{\(\mathit{ff} (\caseof{e_1}{T}{e_2}{e_3}) =\)}                                                                                              \\
				                                                         &       & \(\mathit{ff}(e_1) \cup (\mathit{ff}(e_2) \setminus \mathit{ff}(T)) \cup \mathit{ff}(e_3)\) \\[2mm]
				\(\mathit{ff}((e_1\; e_2))\)                             & \(=\) & \(\mathit{ff}(e_1) \cup \mathit{ff}(e_2)\)                                                  \\[2mm]
				\(\mathit{ff} (x [\links{X}])\)                          & \(=\) & \(\{x/\norm{\links{X}}\}\)                                                                  \\[1mm]
				\(\mathit{ff} (v\,(\links{X}))\)                         & \(=\) & \(\emptyset\)                                                                               \\[1mm]
				\(\mathit{ff} ((\lambda\, x[\links{X}].e) (\links{Y}))\) & \(=\) & \(\mathit{ff}(e) \setminus \{x/\norm{\links{X}}\}\)                                         \\[1mm]
				\(\mathit{ff} ((T_1, T_2))\)                             & \(=\) & \(\mathit{ff}(T_1) \cup \mathit{ff}(T_2)\)                                                  \\[1mm]
				\(\mathit{ff} (\nu X. T)\)                               & \(=\) & \(\mathit{ff} (T)\)                                                                         \\
			\end{tabular}
		\end{center}
		\hrulefill{}
		\caption{Free functors of an expression}\label{table:free-functor}
	\end{figure}
\end{definition}

\subsection{Operational semantics of \(\lambda_{GT}\)}\label{sec:operational-semantics}

First, we define the congruence rules (\(\equiv\))
and the link substitutions, \(T\angled{Y/X}\) and \(G\angled{Y/X}\),
for \(T\) and \(G\)
in the same manner
as we have defined in Section~\ref{sec:hyperlmntal}.
Although there is no graph context in Flat HyperLMNtal,
the link substitution for \(x [\links{X}]\) in \(T\) can be defined
in the same way as the one for atoms in HyperLMNtal.

\subsubsection{Graph Substitution}\label{sub:graph-subsitution}

We define \emph{graph substitution},
which replaces a graph context
whose functor occurs free
by a given subgraph.
The substitution avoids clashes with any bound
functors
by implicit \(\alpha\)-conversion
(capture-avoiding substitution).
Graph substitution is not to
be confused with \emph{hyperlink substitution}.
Intuitively, hyperlink substitution just reconnects hyperlinks.
On the other hand, graph substitution performs deep copying at the semantics level
(though it could or should be implemented with sharing whenever possible).


We define capture-avoiding substitution
\(\theta\)
of a graph context
\(x [\links{X}]\)
with a template \(T\)
in \(e\),
written \(e [T / x [\links{X}]]\),
as in \figref{table:graph-substitution}.
The definition is standard except that it handles the substitution of
the free links of graph contexts in the third rule.

\begin{figure}[t]
	\normalsize
	\hrulefill{}
	\begin{center}
		\begin{tabular}{@{}rc@{}l}
			\((T_1, T_2)\theta\)  & \(=\) & \quad\((T_1 \theta, T_2 \theta)\)                            \\[2mm]
			\((\nu X. T)\theta\)  & \(=\) & \quad\(\nu X. T\theta\)                                      \\[2mm]
			\multicolumn{3}{@{}l}{\((x [\links{X}])[T / y [\links{Y}]]\quad=\)}                          \\
			                      &       & if \(x/\norm{\links{X}} = y/\norm{\links{Y}}\)
			then
			\(T{\angled{\links{X}/\links{Y}}}\)                                                          \\
			                      &       & else
			\(x [\links{X}]\)                                                                            \\[2mm]
			\((C (\links{X}))\theta\)
			                      & \(=\) & \quad\(C (\links{X})\)                                       \\[2mm]
			\multicolumn{3}{@{}l}{%
			\(((\lambda\, x [\links{X}].e) (\links{Z}))[T / y [\links{Y}]]\quad=\)}                      \\
			                      &       & if \(x/\norm{\links{X}} = y/\norm{\links{Y}}\)
			then
			\((\lambda\, x [\links{X}].e) (\links{Z})\)                                                  \\
			                      &       & else if \(x/\norm{\links{X}} \notin \mathit{ff}(e)\) then
			\((\lambda\, x [\links{X}].e[T / y [\links{Y}]]) (\links{Z})\)                               \\
			                      &       & else
			\((\lambda\, z[\links{X}].e
			[z [\links{X}] / x [\links{X}]]
			[T / y [\links{Y}]])
			(\links{Z})\)                                                                                \\
			                      &       & \hfill{} where \(z/\norm{\links{X}} \notin \mathit{ff}(e)\). \\[2mm]
			\multicolumn{3}{@{}l}{%
			\((\caseof{e_1}{T}{e_2}{e_3})\theta\)}                                                       \\
			                      & \(=\) & \quad\(\caseof{e_1\theta}{T}{e_2\theta}{e_3\theta}\)         \\[2mm]
			\((T_1\; T_2)\theta\) & \(=\) & \quad\((T_1 \theta\;\; T_2 \theta)\)                         \\
		\end{tabular}
	\end{center}
	\hrulefill{}
	\caption{Graph Substitution}\label{table:graph-substitution}
\end{figure}

\subsubsection{Matching}\label{sec:matching}


We say that \(T\) matches a graph \(G\) if
there exists graph substitutions \(\theta\) such that
\(G\equiv T\thetas\).
	{
		The graphs in the range of substitutions should not contain free occurence of graph contexts:
		i.e., the substitution should be ground.
	}
Since the matching of \(\lambda_{GT}\)
does not involve abstractions (by Def.~\ref{def:no-abstraction}),
in which case \(G\) of \(\lambda_{GT}\) is essentially the same as
\(G\) of HyperLMNtal, we employ the \(\equiv\) defined in
\figref{table:hyperlmntal-equiv}.

Note that the matching of \(\lambda_{GT}\)
is not subgraph matching (as is standard in graph rewriting systems)
but the matching with the entire graph \(G\) (as is standard in
pattern matching of functional languages).
%
For this reason, the free link names appearing in a template
\(T\)
must exactly match the free links in the graph \(G\) to be matched.
This is to be contrasted with
free links of HyperLMNtal rules that are
effectively \(\alpha\)-convertible
since the rules can match subgraphs by supplementing fusion atoms (\Cite[Section 4.4]{sano2021}).


	{
		The matching can be done non-deterministic.
		We are planning to put constraints over the graph templates in case expressions
		to ensure deterministic matching but it is a future task.
	}

\subsubsection{Reduction}\label{sec:reduction}

We choose the call-by-value evaluation strategy.
	{
		The reason we did not choose call-by-need (or call-by-name) is to avoid
		infinite graphs to use infinite-descent in the verification later in \Cref{sec:autoverify}.
	}

In order to define the small-step reduction relation,
we extend the syntax with evaluation contexts defined as follows:
\[E ::= [] \;|\; (\caseof{E}{T}{e}{e}) \;|\; (E\; e) \;|\; (G\; E) \;|\; T\]
As usual, \(E[e]\) stands for \(E\) whose hole is filled with \(e\).

We define the reduction relation in \figref{table:lgt-reduction}.

\begin{figure}[t]
	\normalsize
	\hrulefill{}
	\vspace{-4mm}
	\begin{prooftree}
		\AxiomC{\(G\equiv T \thetas\)}
		\RightLabel{\small Rd-Case1}
		\UnaryInfC{\((\caseof{G}{T}{e_2}{e_3}) \reduces e_2 \thetas\)}
	\end{prooftree}

	\medskip

	\begin{prooftree}
		\AxiomC{\(\neg\exists \thetas. G\equiv T \thetas\)}
		\RightLabel{\small Rd-Case2}
		\UnaryInfC{\((\caseof{G}{T}{e_2}{e_3}) \reduces e_3\)}
	\end{prooftree}

	\medskip

	\begin{prooftree}
		\AxiomC{\(\mathit{fn}(G) = \{\links{X}\}\)}
		\RightLabel{\small Rd-\(\beta\)}
		\UnaryInfC{\(
			((\lambda\, x [\links{X}].e) (\links{Y})\: G)
			\reduces e [G / x[\links{X}]]
			\)}
	\end{prooftree}

	\medskip

	\begin{prooftree}
		\AxiomC{\(e \reduces e'\)}
		\RightLabel{\small Rd-Ctx}
		\UnaryInfC{\(E[e] \reduces E[e']\)}
	\end{prooftree}
	\hrulefill{}
	\caption{Reduction relation of \(\lambda_{GT}\)}\label{table:lgt-reduction}
\end{figure}

\begin{definition}[Abbreviation rules for \(\lambda\)-abstraction atom]

	We introduce a shorthand notation similar to the \(\lambda\)-calculus.

	\begin{enumerate}
		\item
		      Application is left-associative.

		\item
		      \((\lambda\, x[\links{X}]. (\lambda\, y[\links{Y}].e) (\links{Z})) (\links{Z})\)
		      can be abbreviated as\\
		      \((\lambda\, x[\links{X}]\, y[\links{Y}].e) (\links{Z})\).

		\item
		      {
		      \(((\lambda\, x[\links{X}]. e_1)(\links{Y})\: e_2)\)
		      can be abbreviated as
		      \(\letin{x[\links{X}]}{e_2}{e_1}\).
		      The \(\links{Y}\) will disappear
		      immediately after evaluating the expression, doing nothing,
		      in \(\beta\)-reduction.
		      Thus, we omit the links in the abbreviation.

		      }
	\end{enumerate}
\end{definition}

For example, we can describe a program to append two
singleton difference lists as follows
(detailed description of difference lists will be given in
Section~\ref{ex:prod-dlist}):
\[\begin{array}{l}
		\mathbf{let}\ \mathit{append}[Z] =         \\
		\hspace{1em} (\lambda\, x[Y, X]\; y[Y, X]. \\
		\hspace{3em} x[y[Y], X]                    \\
		\hspace{1em} ) (Z)                         \\
		\mathbf{in}\ \mathit{append}[Z]
		\hspace{0.5em} \mathrm{Cons} (1, Y, X)
		\hspace{0.5em} \mathrm{Cons} (2, Y, X)     \\
	\end{array}\]
We show the whole process of reduction of this program in \figref{fig:dlist-reduction} and graphically in \figref{fig:dlist-reduction-graph}.

Firstly, the \(\lambda\)-abstraction atom is bound
to the graph context \(\mathit{append}[Z]\)\footnote{%
	It may appear that the \(Z\) of \(\mathit{append}[Z]\)
	does not play any role in this example.
	However, such a link becomes necessary when the
	\(\mathit{append}\) is made to appear in a data structure
	(e.g., as in \(\nu Z.(\mathrm{Cons} (Z, Y, X),
	\mathit{append}[Z])\)).  This is why
	\(\lambda\)-abstraction atoms are allowed to have
	argument links.  Once such a function is accessed and
	\(\beta\)-reduction starts, the role of \(Z\) ends,
	while the free links \emph{inside} the abstraction
	atom start to play key roles.
}.
The bound \(\lambda\)-abstraction atom is a function that takes
two difference lists, both having \(X\) and \(Y\) as free
links, and returns their concatenation also having \(X\) and
\(Y\) as its free links.


\begin{figure}[tb]
	\normalsize
	\small
	\hrulefill{}
	\[\begin{array}{@{}c@{\ }l@{}}
			\multicolumn{2}{l}{\mathbf{let}\ \mathit{append}[Z] =
			(\lambda\, x[Y, X]\; y[Y, X].
			\,x[y[Y], X]
			) (Z)}                                                \\
			\multicolumn{2}{l}{\hspace{1.5em}\mathbf{in}\ \mathit{append}[Z]
				\hspace{0.5em} \mathrm{Cons} (1, Y, X)
				\hspace{0.5em} \mathrm{Cons} (2, Y, X)}
			\\[2mm]
			\reduces & (\lambda\, x[Y, X]\; y[Y, X].
			\, x[y[Y], X]) (Z)
			\hspace{0.5em} \mathrm{Cons} (1, Y, X)
			\hspace{0.5em} \mathrm{Cons} (2, Y, X)
			\\[2mm]
			\reduces &
			(\lambda\, y[Y, X].
			\, x[y[Y], X])(Z)[\mathrm{Cons} (1, Y, X)/x[Y, X]]
			\hspace{0.5em} \mathrm{Cons} (2, Y, X)
			\\[0.5mm]
			=        & (\lambda\, y[Y, X].
			\, \mathrm{Cons} (1, y[Y], X))(Z)
			\hspace{0.5em} \mathrm{Cons} (2, Y, X)
			\\[2mm]
			\reduces &
			(\mathrm{Cons} (1, y[Y], X))(X)[\mathrm{Cons} (2, Y, X)/y[Y, X]]
			\\[2mm]
			=        & \mathrm{Cons} (1, \mathrm{Cons} (2, Y), X)
		\end{array}\]
	\hrulefill{}
	\caption{An example of reduction: append operation on difference lists}\label{fig:dlist-reduction}
\end{figure}
\begin{figure}[t]
	\centering
	\includegraphics[width=.9\hsize]{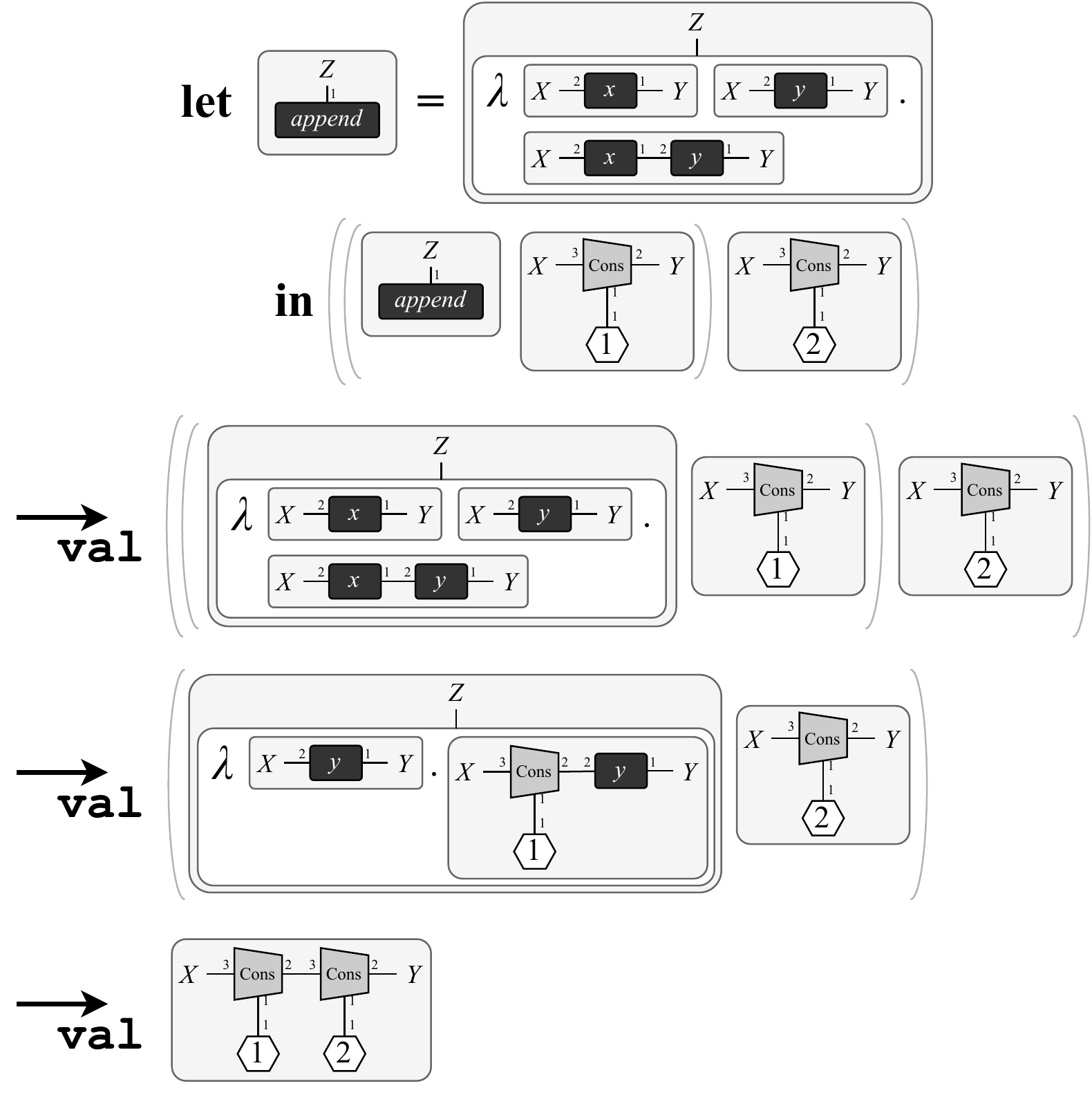}
	\caption{Visualized version of the reduction process in \figref{fig:dlist-reduction}.\\
		Small numbers around a non-unary atom indicate the ordering of arguments,
		and a small dot among edges stands for a \emph{fusion} of the edges.}%
	\label{fig:dlist-reduction-graph}
\end{figure}

A program that pops the last element of a difference list can be described as follows.
\[\begin{array}{l}
		\mathbf{let}\ \mathit{pop}[Z] =                             \\
		\hspace{1em} (\lambda\, x[Y, X].                            \\
		\hspace{2em} \mathbf{case}\ x[Y, X]\ \mathbf{of}            \\
		\hspace{3em} y[\mathrm{Cons} (z, Y), X] \rightarrow y[Y, X] \\
		\hspace{2.5em} |\ \mathbf{otherwise} \rightarrow x[Y, X]    \\
		\hspace{1em} ) (Z)                                          \\
		\mathbf{in}\ \mathit{pop}[Z]
		\hspace{0.5em} \mathrm{Cons} (1, \mathrm{Cons} (2, Y), X)
	\end{array}\]
This will result in
\(\mathrm{Cons} (1, Y, X)\).


j\section{Type System}\label{sec:type-synsem}

In this section, we propose a type system,
\(F_{GT}\),
for the \(\lambda_{GT}\) language.
We define the type of graphs using graph grammar.
This can be regarded as an extension of regular tree grammar,
on which algebraic data types are based.

\subsection{Syntax and rules for \(F_{GT}\)}

Let \(\alpha\) be a syntactic category denoting the identifier of
a type name.
The syntax of types is given in \figref{table:type-syntax}.
It can be observed that the definition of a type employs both
inductive definition (standard in programming languages) and
production rules (standard in formal grammar).
The reason for doing so is that, unlike ADTs, types of graphs cannot be defined inductively in general.
Thus we employed
generative grammar as a well-established formalism for defining
graphs.
Integrating it into \(F_{GT}\) is the research question of the
present work.

\begin{figure}[t]
	\normalsize
	\hrulefill{}
	\begin{center}
		\begin{tabular}{rcll}
			\multicolumn{3}{@{}l}{Atom Name for types}                                                       \\
			\(\tau\)  & \(::=\) & \(\alpha\)                                            & Type Variable      \\
			          & \(|\)   & \(\tau\, (\links{X}) \rightarrow \tau\, (\links{X})\) & Arrow              \\[3mm]
			\\
			\multicolumn{3}{@{}l}{RHS of production rules}                                                   \\
			\(\Tauu\) & \(::=\) & \(\tau\, (\links{X})\)                                & Type Atom          \\
			          & \(|\)   & \(C (\links{X})\)                                     & Constructor Atom   \\
			          & \(|\)   & \(X \bowtie Y\)                                       & Fusion             \\
			          & \(|\)   & \((\Tauu, \Tauu)\)                                    & Molecule           \\
			          & \(|\)   & \(\nu X. \Tauu\)                                      & Hyperlink creation \\[3mm]
			\\
			\multicolumn{3}{@{}l}{Production Rule}                                                           \\
			\(r\)     & \(::=\) & \(\alpha (\links{X}) \longrightarrow \Tauu\)          & Production rule    \\
		\end{tabular}
	\end{center}

	\hrulefill{}
	\caption{Syntax of \(F_{GT}\)}\label{table:type-syntax}
\end{figure}

We extend the \(\lambda\)-expression \(\lambda\, x[\links{X}].e\)
with type annotation
\(\tau\, (\links{X})\)
as
\(\lambda\, x[\links{X}]: \tau\, (\links{X}).e\).

\begin{definition}[Abbreviation rule for an arrow atom]

	We introduce a shorthand notation similar to an arrow in the typed \(\lambda\)-calculus, that is,

	\[(\tau_1 (\links{X})
		\rightarrow
		(\tau_2 (\links{Y}) \rightarrow \tau_3 (\links{Z})) (\links{W})) (\links{W})\]
	can be abbreviated as
	\[(\tau_1 (\links{X})
		\rightarrow
		\tau_2 (\links{Y}) \rightarrow \tau_3 (\links{Z})) (\links{W}).\]
\end{definition}

\begin{definition}[Syntactic constraints]

	A production rule
	\(\alpha (\links{X}) \longrightarrow \Tauu\)
	should satisfy
	\(\mathit{fn} (\Tauu) = \{\links{X}\}\).

\end{definition}

Let \(\Gamma\) be a \emph{typing context} which is a
set of the form
\(x\, [\links{X}]: \tau\, (\links{X})\),
where the \(x\)'s are mutually distinct
and \(t\) should be a type variable or an arrow.
The typing relation
\((\Gamma, P) \vdash e: \tau\, (\links{X})\)
denotes that
\(e\) has the type \(\tau\, (\links{X})\)
under the
type environment \(\Gamma\) and
a set \(P\) of production rules,
whose typing rules are defined as follows.

\begin{definition}[Rules for \({F}_{GT}\)]

	Typing rules for \(F_{GT}\) is given in \figref{fig:fgt-rules}.

	\begin{figure}[t]
		\normalsize
		\hrulefill{}
		\vspace{0.2cm}

		Ty-App
		\vspace{-0.05cm}
		\begin{prooftree}
			\AxiomC{\((\Gamma, P) \vdash e_1 : (\tau_1 (\links{X}) \rightarrow \tau_2 (\links{Y})) (\links{Z})\)}
			\AxiomC{\((\Gamma, P) \vdash e_2 : \tau_1 (\links{X})\)}
			\BinaryInfC{\((\Gamma, P) \vdash (e_1\; e_2) : \tau_2 (\links{Y})\)}
		\end{prooftree}
		\medskip

		Ty-Arrow
		\vspace{-0.1cm}
		\begin{prooftree}
			\AxiomC{\(((\Gamma, x[\links{X}]: \tau_1 ({\links{X}})), P)
				\vdash e : \tau_2 (\links{Y})\)}
			\UnaryInfC{\((\Gamma, P) \vdash
				(\lambda\, x[\links{X}] : \tau_1 ({\links{X}}).e)(\links{Z})
				: (\tau_1 ({\links{X}}) \rightarrow \tau_2 (\links{Y})) (\links{Z})\)}
		\end{prooftree}
		\medskip

		Ty-Var
		\vspace{-0.1cm}
		\begin{prooftree}
			\AxiomC{}
			\UnaryInfC{\((\Gamma\{x[\links{X}]: \tau\, ({\links{X}})\}, P)
				\vdash x[\links{X}] : \tau\, ({\links{X}})\)}
		\end{prooftree}
		\medskip

		Ty-Cong
		\vspace{-0.1cm}
		\begin{prooftree}
			\AxiomC{\((\Gamma, P) \vdash T : \tau\, (\links{X})\)}
			\AxiomC{\(T \equiv T'\)}
			\BinaryInfC{\((\Gamma, P) \vdash T' : \tau\, (\links{X})\)}
		\end{prooftree}
		\medskip

		Ty-Alpha
		\vspace{-0.1cm}
		\begin{prooftree}
			\AxiomC{\((\Gamma, P) \vdash T : \tau\, (\links{X})\)}
			\UnaryInfC{\((\Gamma, P) \vdash T\angled{Z/Y} : \tau\, (\links{X}) \angled{Z/Y}\)}
		\end{prooftree}
		\begin{quote}
			\hspace{1cm} where \(Z \notin \mathit{fn} (T)\)
		\end{quote}
		\medskip
		\vspace{0.2cm}

		Ty-Prod
		\vspace{-0.1cm}
		\begin{prooftree}
			\AxiomC{\((\Gamma, P) \vdash T_1 : \tau_1 (\links{X_1})\)}
			\AxiomC{\hspace{-12pt}\(\dots\)\hspace{-12pt}}
			\AxiomC{\((\Gamma, P) \vdash T_n : \tau_n (\links{X_n})\)}
			\TrinaryInfC{\((\Gamma, P\{\alpha (\links{X}) \longrightarrow \Tauu\})
				\vdash \Tauu [T_1/\tau_1 (\links{X_1}), \dots ,T_n/\tau_n (\links{X_n})] : \alpha (\links{X}) \)}
		\end{prooftree}
		\begin{quote}
			where
			\(\tau_i (\links{X_i})\) are all the type atoms appearing in \(\Tauu\)
		\end{quote}
		\medskip
		\vspace{0.2cm}

		Ty-Case
		\vspace{-0.5cm}
		\begin{prooftree}
			\def\defaultHypSeparation{\hskip -4pt}
			\AxiomC{\((\Gamma, P) \vdash e_1 : \tau_1 (\links{X})\)}
			\AxiomC{\(((\Gamma, \Gamma'), P) \vdash e_2 : \tau_2 (\links{Y})\)}
			\AxiomC{\((\Gamma, P) \vdash e_3 : \tau_2 (\links{Y})\)}
			\TrinaryInfC{\((\Gamma, P) \vdash (\caseof{e_1}{T}{e_2}{e_3}) : \tau_2 (\links{Y})\)}
		\end{prooftree}
		\vspace{1em}
		\hrulefill{}
		\caption{Typing rules for \(F_{GT}\)}\label{fig:fgt-rules}
	\end{figure}

\end{definition}


Ty-App, Ty-Arrow, and Ty-Var are essentially the same as other
functional languages
except that the type of \(F_{GT}\) is written as an atom with free links.
Ty-Var gets the type of the variable from the type environment.
Ty-Cong incorporates the structural congruence rules.
Ty-Alpha \(\alpha\)-converts the free link names of both the graph and
its type.
This rule corresponds to the fact that the free link names in the
rules of HyperLMNtal are
(theoretically) \(\alpha\)-convertible.
Ty-Prod incorporates production rules to the type system.
Ty-Case is also defined in the same manner as in other functional languages,
where \(\Gamma'\) is a type environment that maps types
from all the graph contexts
	{
		appearing
	}
in \(T\),
which we will describe in detail in
\Cref{sub:dynamic-checking}.

\subsection{Examples}\label{sub:graph-types}

In this section,
we introduce some of the production rules,
which we believe describes many of the types of the data structures
for programming in practice.

\begin{example}[Type of a natural number]

	The type of a natural number connected to a free link \(X\)
	can be denoted as \(\mathit{nat}\,(X)\),
	where the production rules are follows.
	\[\begin{array}{l@{~~}c@{~~}l}
			\mathit{nat}\,(X) & \longrightarrow & \mathrm{Zero} (X)               \\
			\mathit{nat}\,(X) & \longrightarrow & \mathrm{Succ} (\mathit{nat}, X)
		\end{array}\]
	(Recall that the RHS of the latter rule is a shorthand of
	\(\nu N.\mathrm{Succ} (N, X),\mathit{nat}\,(N) \) (Def.~\ref{def:abbreviation}).)

	Algebraic data types (ADTs)
	can be easily expressed in the same way as in this example:
	our language and the type system is a natural extension of functional languages and their type systems.

\end{example}

\begin{example}[Type of a difference list]\label{ex:prod-dlist}

	The \(\lambda_{GT}\) language can handle some
	data structures that algebraic data types cannot handle.
	A difference list can be understood as a list with an
	additional link to the last element.
	This is a popular data structure since the early days of logic
	programming in which the links are represented as logical variables.
	It allows us to append two lists in constant time.
	In functional programming, a difference list can be implemented
	using a higher-order function that receives a subsequent list
	and returns the entire list, but we wish to represent such data
	structures in the first-order setting.

	The production rules for a difference list can be defined as follows.
	\[\begin{array}{l@{~~}c@{~~}l}
			\mathit{nodes}\,(Y, X) & \longrightarrow & X \bowtie Y                                          \\
			\mathit{nodes}\,(Y, X) & \longrightarrow & \mathrm{Cons} (\mathit{nat}, \mathit{nodes}\,(Y), X) \\
		\end{array}\]

\end{example}

\begin{example}[Typing a difference list with functions]\label{ex:typing-dlist}

	Since \(\lambda_{GT}\) and its type system \(F_{GT}\) treat functions as first-class citizens,
	it is even possible to have a difference list with functions as its elements.
	\Figref{fig:dlist-typing-example} shows that the graph
	\(G = \mathrm{Cons} (succ, Y, X)\)
	has type
	\(\mathit{nodes}\,(Y, X)\)
	under type environment
	\(\Gamma = succ [Z_1]: (\nattonat) (Z_1)\)
	and production rules \(P=\{P_1,P_2\}\) where
	\[ \begin{array}{l@{~~}c@{~~}l@{~~}c@{~~}l}
			\mathit{nodes}\,(Y, X) & \longrightarrow & X \bowtie Y                                       & \cdots & P_1 \\
			\mathit{nodes}\,(Y, X) & \longrightarrow & \mathrm{Cons} (\nattonat, \mathit{nodes}\,(Y), X) & \cdots & P_2 \\
		\end{array}\]

	\begin{figure*}[t]
		\normalsize
		\small
		\begin{prooftree}
			\def\defaultHypSeparation{\hskip .1in}
			\AxiomC{}
			\RightLabel{Ty-Var}
			\UnaryInfC{\((\Gamma\{succ[Z_1]: (\nattonat) (Z_1)\}, P) \vdash succ[Z_1] : (\nattonat) (Z_1)\)}

			\AxiomC{}
			\RightLabel{Ty-Prod}
			\UnaryInfC{\((\Gamma, P\{P_1\}) \vdash X \bowtie Y : \mathit{nodes}\,(Y, X)\)}
			\RightLabel{Ty-Alpha}
			\UnaryInfC{\((\Gamma, P) \vdash Z_2 \bowtie Y : \mathit{nodes}\,(Z_2, X)\)}

			\RightLabel{Ty-Prod}
			\BinaryInfC{\((\Gamma, P\{P_2\}) \vdash G' : \mathit{nodes}\,(Y, X)\)
				\hspace{0.5em} where \hspace{0.5em}
				\(G' = \nu Z_1 Z_2. (\mathrm{Cons} (Z_1, Z_2, X), succ [Z_1], Z_2 \bowtie Y)\)
			}

			\RightLabel{Ty-Cong}
			\AxiomC{\hspace{-10pt}\(G \equiv G'\)}

			\BinaryInfC{\((\Gamma, P) \vdash G : \mathit{nodes}\,(Y, X)\)
				\hspace{0.5em} where \hspace{0.5em}
				\(G = \mathrm{Cons} (succ, Y, X)\)
			}
		\end{prooftree}
		\caption{Type checking a difference list}\label{fig:dlist-typing-example}
	\end{figure*}

\end{example}

\begin{example}[Type of a doubly-linked difference list]\label{ex:prod-dbl}

	A doubly-linked difference list is a list with four free
	links, two different links for each end.
	Although the (hyper)links of \(\lambda_{GT}\) and HyperLMNtal
	are undirected, we are interested in using them to model
	directed hyperlinks (roughly corresponding to pointers in
	imperative languages) that are to be `followed' in one direction.
	As with difference lists, the addition of elements to the tail of
	the list can be done in constant time, as desired in
	representing deques.
	%
	Of course, doubly-linked lists that are not difference lists can also be handled in an obvious way.
	\[\begin{array}{l@{~~}c@{~~}l}
			\mathit{nodes}\,(F', B, B', F) & \longrightarrow & F \bowtie B,B' \bowtie F'                                                 \\
			\mathit{nodes}\,(F', B, B', F) & \longrightarrow & \nu X. \mathrm{Cons} (\mathit{nat}, F', \mathit{nodes}\,(X, B, B'), X, F) \\
		\end{array}\]
\end{example}

\begin{example}[Type of difference skip lists]\label{ex:prod-skip}

	By extending the type definition of difference lists, the type of unbounded-level skip lists can be defined.
	This implies that we can also define a type for skip lists
	with a nil node at the end and/or whose level is fixed.
	\[\begin{array}{l@{~~}c@{~~}l}
			\mathit{nodes}\,(Y, X) & \longrightarrow & X \bowtie Y                                                             \\
			\mathit{nodes}\,(Y, X) & \longrightarrow & \mathrm{Cons} (\mathit{nat}, \mathit{forks}\,(Y), X)                    \\
			\mathit{forks}\,(Y, X) & \longrightarrow & \mathrm{Next} (\mathit{nodes}\,(Y), X)                                  \\
			\mathit{forks}\,(Y, X) & \longrightarrow & \nu Z.\mathrm{Fork} (Z, \mathit{forks}\,(Z), X), \mathit{nodes}\,(Y, Z) \\
		\end{array}\]
	We also show the visualized version of production rules in \figref{fig:skiplist-prodrule}
	and an example difference skip list in \figref{fig:skiplist-example}.
	\begin{figure}[t]
		\centering
		\includegraphics[width=.8\hsize]{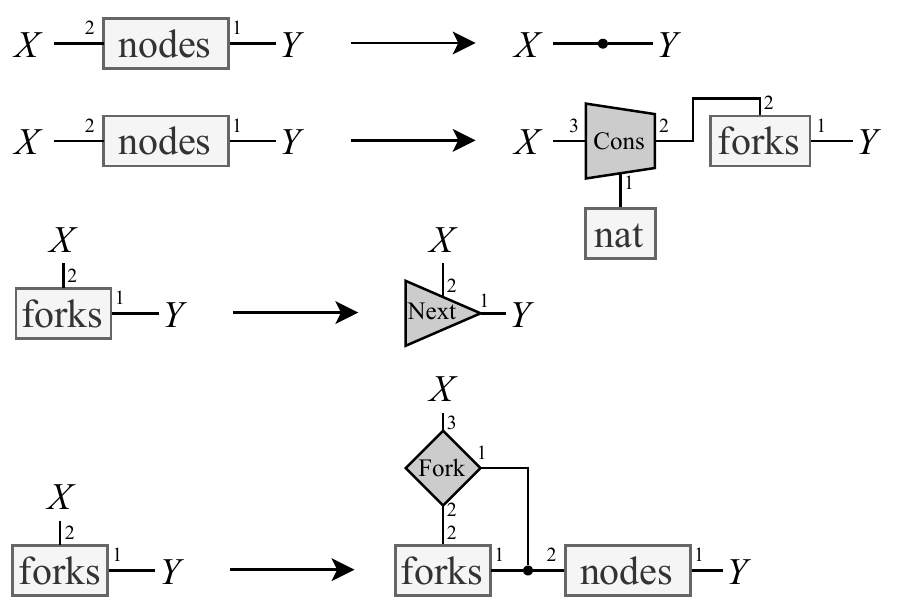}
		\caption{The production rules for the type of difference skip list}
		\label{fig:skiplist-prodrule}
	\end{figure}
	\begin{figure}[t]
		\centering
		\includegraphics[width=.8\hsize]{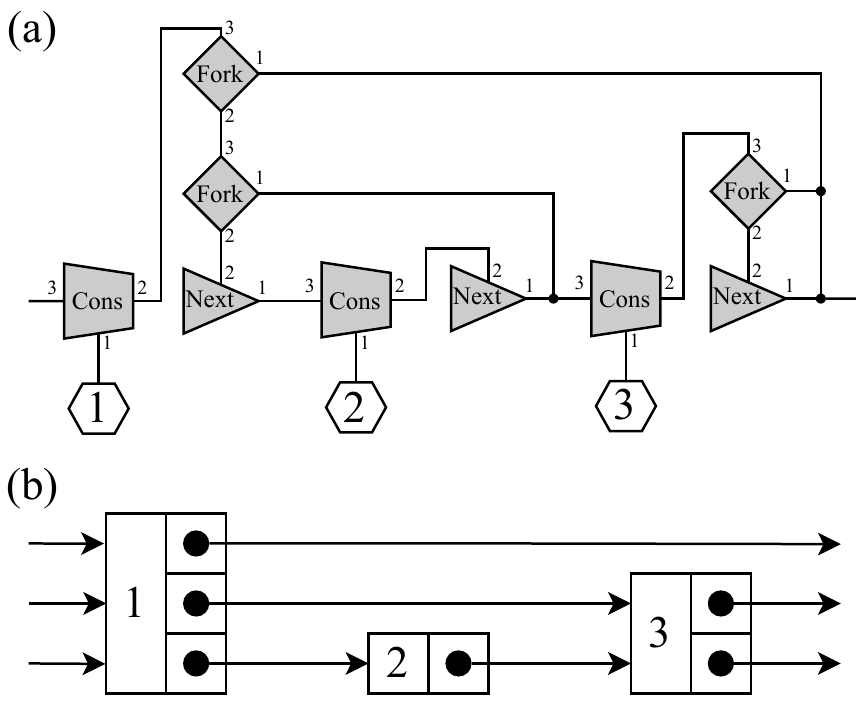}
		\caption{An example of a skip list:\\ (a) in our framework, (b) with pointers}
		\label{fig:skiplist-example}
	\end{figure}
\end{example}

\begin{example}[Type of a leaf-linked tree]\label{ex:prod-lltree}

	A leaf-linked tree is a graph with three free links (say
	\(X,L,R\)) which is a tree whose root is represented by \(X\) and
	whose leaves form a difference list represented by \(L\) and \(R\).
	\[\begin{array}{l@{~~}c@{~~}l}
			\mathit{lltree}\,(L, R, X) & \longrightarrow & L \bowtie X, \mathrm{Leaf} (\mathit{nat}, R, X)                            \\
			\mathit{lltree}\,(L, R, X) & \longrightarrow & \nu Y. \mathrm{Node} (\mathit{lltree}\,(L, Y), \mathit{lltree}\,(Y, R), X) \\
		\end{array}\]
\end{example}

\begin{example}[Type of a threaded tree]\label{ex:prod-thtree}

	A threaded tree is somewhat similar to a leaf-linked tree but each
	non-terminal node has access to the the rightmost leaf the left
	subtree and the leftmost leaf of the right subtree.
	\[\begin{array}{l@{~~}c@{~~}l}
			\mathit{thtree}\,(L, R, X) & \longrightarrow & L \bowtie X, \mathrm{Leaf} (\mathit{nat}, R, X) \\
			\mathit{thtree}\,(L, R, X) & \longrightarrow & \mathrm{Node}
			(\mathit{thtree}\,(L, X), \mathit{thtree}\,(X, R), X)                                          \\
		\end{array}\]
\end{example}

\subsection{Properties of \(F_{GT}\)}\label{fgt-properties}


This section discusses some properties of \(\lambda_{GT}\) and
\(F_{GT}\).  As mentioned in Section~\ref{sec:syntax-semantics},
we keep the language small to focus on the handling of
graph structures, more specifically the handling of graphs by
pattern matching with graph contexts.
In particular, it has no explicit mechanism
(such as \(\mathbf{let\ rec}\) or \(\mathbf{fix}\))
to deal with recursive functions.
This is because those features can be achieved essentially in the same
way as other functional languages do.

\subsubsection{Soundness of \(F_{GT}\)}\label{sub:soundness}

\begin{lemma}[Progress]\label{lem:progress}
	If
	\((\emptyset, P) \vdash e: \tau\, (\links{X})\),
	then
	\(e\) is a value or \(\exists e'. e \reduces e'\).
\end{lemma}

\begin{proof}
	By induction on the derivation of
	\((\emptyset, P) \vdash e: \tau\, (\links{X})\).
	Notice that
	the only new extension from other functional languages
	in expressions (\figref{table:lgt-syntax})
	is Case,
	and the Case expression is never stuck because if matching fails;
	it just branches to \textbf{otherwise} and evaluation proceeds.
\end{proof}

\begin{lemma}[Substitution]\label{lem:subst}
	If
	\[(\Gamma, P) \vdash e_1: \tau_1 (\links{Y_1})\]
	and
	\[((\Gamma, x[\links{X_1}]: \tau_1 (\links{Y_1})), P)
		\vdash e_2: \tau_2 (\links{Y_2})\]
	then
	\[(\Gamma, P) \vdash e_2[e_1 / x[\links{X_1}]] : \tau_2 (\links{Y_2}).\]
\end{lemma}

\begin{proof}
	By induction on the derivation of
	\(((\Gamma, x[\links{X_1}]: \tau_1 (\links{Y_1})), P)
	\vdash e_2: \tau_2 (\links{Y_2})\).
\end{proof}

\begin{lemma}[Preservation]\label{lem:preserve}
	If
	\((\Gamma, P) \vdash e: \tau\, (\links{X})\)
	and
	\(e \reduces e'\),
	then
	\((\Gamma, P) \vdash e' : \tau\, (\links{X})\).
\end{lemma}

\begin{proof}
	Proved using the \Cref{lem:subst}.
\end{proof}

\begin{theorem}[Soundness]\label{th:soundness}
	If
	\((\emptyset, P) \vdash e: \tau\, (\links{X})\),
	and
	\(e \reduces^{\ast} e'\)
	then
	\(e'\) is a value or \(\exists e''. e' \reduces e''\).
\end{theorem}

\begin{proof}
	Follows from \Cref{lem:progress} and \Cref{lem:preserve}.
\end{proof}

\subsubsection{Relation with graph reduction}\label{sub:fgt-reduction}

Structured Gamma\cite{structuredgamma} is a
first-order
graph rewriting system developed to represent and reason about
the shapes of pointer data structures.  The framework of Structured
Gamma was then adapted to LMNtal (whose graph structures are dual to
those of Structured Gamma, roughly speaking) to design and implement
LMNtal ShapeType\cite{gcm2021}.  Despite several syntactic
variations (such as the duality of nodes/links and the presence/absence of
hyperlinks), Structured Gamma and LMNtal ShapeType can (essentially)
handle graphs
of \(\lambda_{GT}\) without \(\lambda\)-abstraction atoms.
The typing relation {\`a} la Structured Gamma
and LMNtal ShapeType is defined as follows.

\begin{definition}[Typing relation in Structured Gamma/LMNtal ShapeType]\label{def:structured-gamma}

	\(P \vdash \tau: \alpha (\links{X})\)
	iff
	\(\alpha (\links{X}) \rightsquigarrow^{\ast}_{P} \tau\)
	and
	\(\tau\) does not contain type variables or arrow atoms

\end{definition}

We have shown that the typing relation in our type system \(F_{GT}\)
subsumes
the one in Structured Gamma in the following sense.


\begin{theorem}[\(F_{GT}\) and HyperLMNtal reduction]\label{th:fgt-stgamma}
	\[\begin{array}{l}
			(\Gamma, P)
			\vdash
			T : \tau\, (\links{X}) \\
			\qquad\Leftrightarrow
			\tau\, (\links{X})
			\rightsquigarrow^{\ast}_{P}
			T
			{\overrightarrow{[\tau_i (\links{Y_i}) / x_i [\links{X_i}]]}}^{i}
			{[\overrightarrow{\tau_i (\links{Z_i}) / {(\lambda \dots)}_i (\links{W_i})}]}^{i}
		\end{array}\]

	where
	\begin{itemize}
		\item
		      \(\Gamma =
		      {\overrightarrow{x_i [\links{X_i}]: \tau_i (\links{X_i})}}^{i}
		      \),

		\item
		      \({(\lambda \dots)}_{i} (\links{W_i})\)
		      are all the \(\lambda\)-abstraction atoms in \(T\), and
		\item
		      \((\Gamma, P) \vdash {(\lambda \dots)}_i (\links{W_i}) : \tau_i (\links{Z_i})\)

	\end{itemize}

\end{theorem}

\begin{proof}
	For \(\Rightarrow\),
	we can prove by induction on the last applied \(F_{GT}\) rules.
	For \(\Leftarrow\),
	We prove by induction on the length of the reduction
	\(\rightsquigarrow^{\ast}_{P}\).
\end{proof}

Note that if no graph contexts
or \(\lambda\)-expressions appear in \(T\),
by \Cref{th:fgt-stgamma},
the typing relation in \(F_{GT}\) is equivalent to the one in Structured Gamma.
In other words, our type system is an extension of Structured Gamma
to allow graph contexts and \(\lambda\)-abstraction atoms.
This allows us to take advantage of research results on Structured Gamma,
its derivative LMNtal ShapeType,
and parsing of graphs using graph grammar.

\begin{example}[
		Theorem~\ref{th:fgt-stgamma} on the difference list example]

	Here, we see that \Cref{th:fgt-stgamma} holds on
	\Cref{ex:typing-dlist}.
	Recall that
	\((succ [Z_1]: (\nattonat) (Z_1), P) \vdash \mathrm{Cons}\allowbreak (succ,\allowbreak Y,\allowbreak X) : \mathit{nodes}\,(Y, X)\)
	holds in \(F_{GT}\),
	which can also be shown using HyperLMNtal reduction as follows.
	\[\begin{array}{l}
			\mathit{nodes}\,(Y, X)                                                                  \\
			\rightsquigarrow_{P_2}
			\nu Z_1 Z_2. (\mathrm{Cons} (Z_1, Z_2, X), (\nattonat) (Z_1), \mathit{nodes}\,(Y, Z_2)) \\
			\rightsquigarrow_{P_1}
			\nu Z_1 Z_2. (\mathrm{Cons} (Z_1, Z_2, X), (\nattonat) (Z_1), Z_2 \bowtie Y)            \\
			\equiv
			\nu Z_1. (\mathrm{Cons} (Z_1, Y, X), (\nattonat) (Z_1))                                 \\
			=
			\mathrm{Cons} (succ , Y, X) [(\nattonat) (Z_1) / succ [Z_1]]
		\end{array}\]

\end{example}

\subsection{Type checking case expressions}\label{sub:dynamic-checking}
{%

	\(\lambda_{GT}\) allows pattern matching of graphs.
	In pattern matching, graph contexts can be used as wildcards.
	Since a graph context can match any graph
	as long as the sets of free links are the same,
	we cannot naively ensure that the type of the graph matches
	the intended type of the context.
	Therefore, we allow the typing annotation of graph contexts.
}

To allow type annotation in pattern matching,
we extend the syntax of the graph template \(T\).
A type annotation \(T: \tau\, (\links{X})\) ensures that
the type of the graph matched with \(T\) is of type \(\tau\, (\links{X})\).

To evaluate pattern matching with annotations,
we extend the matching mechanism.
\(\matches{G}{T}{\thetas}\)
denotes that (i) the graph context \(T\) can match the
graph \(G\) with graph substitutions \(\thetas\) and that (ii)
each subgraph of \(G\) matched by a subcontext of \(T\) satisfies
the type constraint attached to the subcontext.
\(\matches{G}{T}{\thetas}\)
is defined inductively as in \figref{table:matching}.
It is a straightforward inductive argument
to see that \figref{table:matching} extends
the matching defined in Section~\ref{sec:matching} with
the rule Mt-Ty for type checking.

The type annotations that do not match the type definitions of production rules
could be reported as bugs,
which could be analyzed easily and statically.
%
Also, for simplicity, henceforth
we will assume that all the graph contexts are type-annotated
and make it a future task to support unannotated graph contexts.

	{%

		The matching can be non-deterministic; that is, given
		a graph and a pattern, there may in general be more
		than one way in which graph contexts in the pattern
		are bound to subgraphs.
		However, the non-determinacy of the matching does not affect the soundness of the type system
		since the system proves that every execution path is type-safe.
	}

\begin{figure}[t]
	\normalsize
	\hrulefill{}
	\begin{prooftree}
		\AxiomC{\(\mathit{fn}(G) = \{\links{X}\}\)}
		\RightLabel{Mt-Var}
		\UnaryInfC{\(\matches{G}{x [\links{X}]}{[G / x [\links{X}]]}\)}
	\end{prooftree}
	\medskip

	\begin{prooftree}
		\AxiomC{}
		\RightLabel{Mt-Triv}
		\UnaryInfC{\(\matches{G}{G}{[]}\)}
	\end{prooftree}
	\medskip

	\begin{prooftree}
		\AxiomC{\(\matches{G_1}{T_1}{\overrightarrow{\theta_1}}\)}
		\AxiomC{\(\matches{G_2}{T_2}{\overrightarrow{\theta_2}}\)}
		\RightLabel{Mt-Mol}
		\BinaryInfC{\(\matches{(G_1, G_1)}{(T_1, T_2)}{\overrightarrow{\theta_1} \overrightarrow{\theta_2}}\)}
	\end{prooftree}
	\begin{quote}
		\hspace{0.3cm} where \(\mathrm{dom} ({\overrightarrow{\theta_1}})
		\cap
		\mathrm{dom} ({\overrightarrow{\theta_2}}) = \emptyset\)
	\end{quote}
	\medskip
	\vspace{0.3cm}

	\begin{prooftree}
		\AxiomC{\(\matches{G}{T}{\thetas}\)}
		\RightLabel{Mt-\(\nu\)}
		\UnaryInfC{\(\matches{\nu X.G}{\nu X.T}{\thetas}\)}
	\end{prooftree}
	\medskip

	\begin{prooftree}
		\AxiomC{\(\matches{G_1}{T}{\thetas}\)}
		\AxiomC{\(G_1 \equiv G_2\)}
		\RightLabel{Mt-Cong}
		\BinaryInfC{\(\matches{G_2}{T}{\thetas}\)}
	\end{prooftree}
	\medskip

	\begin{prooftree}
		\AxiomC{\(\matches{G}{T}{\thetas}\)}
		\AxiomC{\(G: \tau\, (\links{X})\)}
		\RightLabel{Mt-Ty}
		\BinaryInfC{\(\matches{G}{(T: \tau\, (\links{X}))}{\thetas}\)}
	\end{prooftree}
	\hrulefill{}
	\caption{Matching with a template and graph substitutions}\label{table:matching}
\end{figure}

The program that pops the last element of a difference list
we have introduced in \Cref{sec:reduction}
can be handled with
type-annotations as follows.
\[
	\begin{array}{l}
		(\Gamma, P)
		\vdash \\
		\begin{array}{l}
			\hspace{0em} (\lambda\, x[Y, X]: \mathit{nodes}\,(Y, X).         \\
			\hspace{1em} \mathbf{case}\ x[Y, X]\ \mathbf{of}                 \\
			\hspace{2em} \nu Z_1. Z_2. (y[Z_1, X]: \mathit{nodes}\,(Z_1, X), \\
			\hspace{5em} \mathrm{Cons} (Z_2, Y, Z_1),                        \\
			\hspace{5em} z[Z_2]: \mathit{nat}\,(Z_2))                        \\
			\hspace{3em} \rightarrow y[Y, X]                                 \\
			\hspace{1.5em} |\ \mathbf{otherwise} \rightarrow x[Y, X]         \\
			\hspace{0em} ) (Z)                                               \\
		\end{array}:
		\begin{array}{l}
			\hspace{0em} (\mathit{nodes}\,(Y, X) \rightarrow \\
			\hspace{1em} \mathit{nodes}\,(Y, X)              \\
			\hspace{0em} ) (Z)                               \\
		\end{array}
	\end{array}
\]
This can be typed using Ty-Case
where the \(\Gamma'\) stands for the annotated typing relations
\(
y[Z_1, X]: \mathit{nodes}\,(Z_1, X),
z[Z_2]: \mathit{nat}\,(Z_2)
\).


\section{Extending the type system}\label{sec:fgt-ext}

In this section, we deal with an example which the type system in
Section~\ref{sec:type-synsem}
fails to verify.
The type system in
Section~\ref{sec:type-synsem}
was actually for \emph{parsing} when dealing with graphs;
it just checks if the graph can be generated from the annotated type
variable atom,
i.e., the start symbol.
Algebraic data types can be handled in this
manner because
they can only be generated according to the grammar that defines the type.
However, in the case of graphs, more powerful operations are possible,
for example the concatenation of difference lists.
In this section, we propose an extended verification framework to deal with such cases.

\subsection{Motivation}

As a running example, we consider a typed version of the
following program for appending two difference lists
introduced in Section~\ref{sec:reduction}.

\[\begin{array}{l}
		\hspace{0em} (\lambda\: x[Y, X]: \mathit{nodes}\,(Y, X) \\
		\hspace{1.1em} y[Y, X]: \mathit{nodes}\,(Y, X).         \\
		\hspace{2em} x[y[Y], X]                                 \\
		\hspace{0em} ) (Z)                                      \\
	\end{array}\]

It seems natural that the following typing relation holds,
where \(\mathit{append} [Z]\) is the \(\lambda\)-abstraction atom above.
\[\begin{array}{l}
		(\Gamma, P) \vdash \\
		\quad \mathit{append}[Z]
		: (\mathit{nodes}\,(Y, X) \rightarrow \mathit{nodes}\,(Y, X) \rightarrow \mathit{nodes}\,(Y, X)) (Z)
	\end{array}\]

However, this program cannot be verified by
directly using the rules in the type system in
Section~\ref{sec:type-synsem}.

\begin{theorem}[]
	The \(\mathit{append}\) operation on difference lists fails to
	verify on the previously defined \(F_{GT}\).
\end{theorem}

\begin{proof}
	We need to prove
	\[\begin{array}{l}
			((x [Y, X]: \mathit{nodes}\,(Y, X),
			y [Y, X]: \mathit{nodes}\,(Y, X)), P) \vdash \\
			\qquad x [y [Y],  X]: \mathit{nodes}\,(Y, X)
		\end{array}\]
	to verify the present example.
	\Cref{th:fgt-stgamma} states that,
	if we can successfully prove the typing relation using \(F_{GT}\),
	we should be able to prove
	\(\mathit{nodes}\,(Y, X) \rightsquigarrow^{\ast}_{P} \mathit{nodes}\,(\mathit{nodes}\,(Y), X)\).
	However,
	applying the production rules of difference lists
	cannot increase the number of \(\mathit{nodes}/2\) atoms.
	Therefore, applying the production rules to the annotated type
	variable atom
	\(\mathit{nodes}\,(Y, X)\) will never yield
	\(\mathit{nodes}\,(\mathit{nodes}\,(Y), X)\).
\end{proof}

However,
it is obvious that
appending two difference lists returns a difference list,
and this operation should be supported.
We extend the previously defined \(F_{GT}\)
to enable such verification.

\subsection{Extension on \(F_{GT}\)}

We start with the first attempt of the extension.

\begin{definition}[Extension on \(F_{GT}\) (Unrefined)]\label{def:fgt-ext-unrefined}

	For a graph template \(T\),
	it is sufficient if the typing succeeds after replacing
	each graph context in \(T\)
	by all possible values of the types attached to the graph
	context, or more formally,
	as in \figref{fig:ty-subst-unrefined}.

	\begin{figure*}[t]

		\normalsize
		\hrulefill{}
		\vspace{0.2cm}

		\begin{prooftree}
			\AxiomC{\(
			\forall \overrightarrow{G_i}. \left(
			\left(
			\bigwedge\limits^{i} \left(
			(\emptyset, P) \vdash G_i : \tau_i (\links{X_i})
			\right)
			\right)
			\Rightarrow
			(\emptyset, P) \vdash
			T
			{\overrightarrow{\left[G_i/x_i [\links{X_i}]\right]}}^{i}
			: \tau\, (\links{X})
			\right)
			\)}
			\RightLabel{Ty-Subst (unrefined)}
			\UnaryInfC{\(
				\left(
				{\overrightarrow{x_i [\links{X_i}] : \tau_i (\links{X_i})}}^{i}
				, P
				\right)
				\vdash T : \tau\, (\links{X})\)}
		\end{prooftree}

		\hspace{2.8cm}
		where \({\overrightarrow{x_i[\links{X_i}]}}^i = \mathit{ff} (T)\)

		\vspace{1em}
		\hrulefill{}
		\caption{Extension on \(F_{GT}\) (unrefined)}\label{fig:ty-subst-unrefined}
	\end{figure*}

\end{definition}

{
The (apparently intuitive) rule in \Cref{def:fgt-ext-unrefined}
has \(\Rightarrow\) on the antecedent
and the typing relation we are going to define
(the parameter of the generating function)
appears on the left-hand side of the \(\Rightarrow\).
Unfortunately, then, we cannot ensure the monotonicity of the
generating function
and the existence of a least fixed point,
which is the typing relation we want to define.

Now we consider how to fix this, which we have found is not trivial
or straightforward.
%
%
If we define a typing relation, say \(R_0\), without Ty-Subst
and define the left-hand side of the \(\Rightarrow\) of Ty-Subst with \(R_0\),
we can ensure the monotonicity of the generating function
and the typing relation becomes well-defined.

First, we prepare two sets of typing rules,
one with all the \({:}\)'s in \figref{fig:fgt-rules} rewritten as \({:}_0\)
and the other as \({:}_1\).
Then, we can define \(R_0\) only with typing rules with \({:}_0\),
which is well-defined.

Next, we define the typing relation, say \(R_1\), using
typing rule with \({:}_1\)
and the rule in \figref{fig:ty-subst-unrefined2} (Ty-Subst with \({:}_0\) and \({:}_1\)).
Since the left-hand side of \(\Rightarrow\) in \figref{fig:ty-subst-unrefined2}
uses the already defined \(R_0\),
it can be interpreted as a monotonic function
and the typing relation \(R_1\) is well-defined.

\begin{figure*}[t]

	\normalsize
	\hrulefill{}
	\vspace{0.2cm}

	\begin{prooftree}
		\AxiomC{\(
		\forall \overrightarrow{G_i}. \left(
		\left(
		\bigwedge\limits^{i} \left(
		(\emptyset, P) \vdash G_i \:{:}_0\: \tau_i (\links{X_i})
		\right)
		\right)
		\Rightarrow (\emptyset, P) \vdash
		T
		{\overrightarrow{\left[G_i/x_i [\links{X_i}]\right]}}^{i}
		\:{:}_1\: \tau\, (\links{X})
		\right)
		\)}
		\RightLabel{Ty-Subst (unrefined 2)}
		\UnaryInfC{\(
			\left(
			{\overrightarrow{x_i [\links{X_i}]
					\:{:}_{1}\: \tau_i (\links{X_i})}}^{i}
			, P
			\right)
			\vdash T \:{:}_1\:  \tau\, (\links{X})\)}
	\end{prooftree}

	\hspace{2.8cm}
	where \({\overrightarrow{x_i[\links{X_i}]}}^i = \mathit{ff} (T)\)

	\vspace{1em}
	\hrulefill{}
	\caption{Extension on \(F_{GT}\) (unrefined 2)}\label{fig:ty-subst-unrefined2}
\end{figure*}

However, we cannot ensure the soundness if we define Ty-Subst in such a way
because the antecedent of the rule ensures the safety if
\(G_i\) has a type \(\tau_i (\links{X_i})\) in \(R_0\)
but
the antecedent does not ensure the safety when
\(G_i\) has a type \(\tau_i (\links{X_i})\) only in \(R_1\).
Since \(R_1\) can handle more programs than \(R_0\),
this may violate the soundness of the system.

For example, consider the case where \(G_i\) is a graph containing a function that concatenates difference lists
and \(\tau_i (\links{X_i})\) contains an arrow type for a function
that takes two difference lists (as curried arguments) and return a difference list.
Since it is not verifiable in \(R_0\) that a function that concatenates difference lists returns a difference list,
we can make the left-hand side of \(\Rightarrow\) on Ty-Subst false.
In such a case,
the antecedent of Ty-Subst is satisfied no matter what the right-hand
side of \(\Rightarrow\) is.
Thus, we cannot ensure safety for the case where
we bound a graph containing a function that concatenates difference lists.
However, since the consequent of Ty-Subst uses \({:}_1\),
it allows the \(x_i [\links{X_i}]\) to be bound to a graph that includes
a function that concatenates difference lists.

%
%
%
%
%

Therefore,
we need a more refined framework that allows
the indices we have
attached to
\({:}\) previously
to be different for each type.
Accordingly, we introduce the notion of \emph{ranks} for types and Ty-Subst.

If the typing on the left-hand side of \(\Rightarrow\) does not use Ty-Subst,
which we will define,
it can be interpreted as a monotonic function and is well-defined.
Therefore, we introduce ranks into Ty-Subst
so that the left-hand side typing of \(\Rightarrow\) can only use Ty-Subst
with a lower rank that has already been defined.

We first introduce ranks to the type.
We denote the type with Rank \(n \: (n \geq 0)\) as \({\tau}^{n} (\links{X})\).
We extend the rules in \figref{fig:fgt-rules} so that the types have ranks.

\begin{definition}[Rules for \({F}_{GT}\) with ranks]

	Typing rules for \(F_{GT}\) with ranks are given in \figref{fig:fgt-rules-ranked}.

	Notice that we have added a new typing rule Ty-Sub, a rule for subtyping.


	\begin{figure}[t]

		\normalsize
		\hrulefill{}
		\vspace{0.2cm}

		Ty-App
		\begin{prooftree}
			\AxiomC{\((\Gamma, P) \vdash e_1
			: {({\tau_1}^{n} (\links{X})
			\rightarrow {\tau_2}^{m} (\links{Y}))}^{m}
			(\links{Z})
			\)}
			\AxiomC{\((\Gamma, P) \vdash e_2 : {\tau_1}^{n} (\links{X})\)}
			\BinaryInfC{\((\Gamma, P) \vdash (e_1\; e_2) :  {\tau_2}^{m} (\links{Y})\)}
		\end{prooftree}
		\medskip

		Ty-Arrow
		\begin{prooftree}
			\AxiomC{\(((\Gamma, x[\links{X}] :  {\tau_1}^{n} (\links{X})), P)
			\vdash e : {\tau_2}^{m} (\links{Y})\)}

			\AxiomC{\(n \leq m\)}

			\BinaryInfC{\((\Gamma, P) \vdash
			(\lambda\, x[\links{X}] : {\tau_1}^{n} (\links{X}).e)(\links{Z})
			: {({\tau_1}^{n} (\links{X})
			\rightarrow {\tau_2}^{m} (\links{Y}))}^{m}
			(\links{Z})
			\)}
		\end{prooftree}
		\medskip

		Ty-Var
		\begin{prooftree}
			\AxiomC{}
			\UnaryInfC{\((\Gamma\{x[\links{X}] : {\tau}^{n}\,(\links{X})\}, P)
			\vdash x[\links{X}] : {\tau}^{n}\,(\links{X})\)}
		\end{prooftree}
		\medskip

		Ty-Cong
		\vspace{-0.1cm}
		\begin{prooftree}
			\AxiomC{\((\Gamma, P) \vdash T : {\tau}^{n}\,(\links{X})\)}
			\AxiomC{\(T \equiv T'\)}
			\BinaryInfC{\((\Gamma, P) \vdash T' : {\tau}^{n} (\links{X})\)}
		\end{prooftree}
		\medskip

		Ty-Alpha
		\vspace{-0.1cm}
		\begin{prooftree}
			\AxiomC{\((\Gamma, P) \vdash T : {\tau}^{n}(\links{X})\)}
			\UnaryInfC{\((\Gamma, P) \vdash T\angled{Z/Y} : {\tau}^{n} (\links{X}) \angled{Z/Y}\)}
		\end{prooftree}
		\begin{quote}
			\hspace{1cm} where \(Z \notin \mathit{fn} (T)\)
		\end{quote}
		\medskip
		\vspace{0.2cm}

		Ty-Prod
		\vspace{-0.1cm}
		\begin{prooftree}
			\AxiomC{\((\Gamma, P) \vdash T_1 : {\tau_1}^{n_1} (\links{X_1})\)}
			\AxiomC{\hspace{-12pt}\(\dots\)\hspace{-12pt}}
			\AxiomC{\((\Gamma, P) \vdash T_n : {\tau_m}^{n_m} (\links{X_m})\)}
			\TrinaryInfC{\((\Gamma, P\{\alpha (\links{X}) \longrightarrow \Tauu\})
			\vdash \Tauu [T_1/\tau_1 (\links{X_1}), \dots ,T_m/\tau_m (\links{X_m})]
			: {\alpha}^{\max n_i} (\links{X}) \)}
		\end{prooftree}
		\begin{quote}
			where
			\(\tau_i (\links{X_i})\) are all the type atoms appearing in \(\Tauu\)
		\end{quote}
		\medskip
		\vspace{0.2cm}

		Ty-Case
		\begin{prooftree}
			\def\defaultHypSeparation{\hskip -2pt}
			\AxiomC{\(
			(\Gamma, P) \vdash e_1 : {\tau_1}^{n} (\links{X})
			\)}
			\AxiomC{\(\begin{array}[b]{r@{\:\:}l}
					((\Gamma, \Gamma'), P) \vdash e_2 & : {\tau_2}^{m} (\links{Y})  \\[1mm]
					(\Gamma, P) \vdash e_3            & : {\tau_2}^{m}  (\links{Y})
				\end{array}\)}
			\AxiomC{\(n \leq m\)}
			\TrinaryInfC{\((\Gamma, P) \vdash (\caseof{e_1}{T}{e_2}{e_3}) : {\tau_2}^{m} (\links{Y})\)}
		\end{prooftree}
		\medskip
		\vspace{0.2cm}

		Ty-Sub
		\vspace{-0.1cm}
		\begin{prooftree}
			\AxiomC{\((\Gamma, P) \vdash e : {\tau}^{n}\,(\links{X})\)}
			\AxiomC{\(n < m\)}
			\BinaryInfC{\((\Gamma, P) \vdash e : {\tau}^{m} (\links{X})\)}
		\end{prooftree}

		\vspace{1em}
		\hrulefill{}
		\caption{Typing rules for \(F_{GT}\) with ranks}\label{fig:fgt-rules-ranked}
	\end{figure}

\end{definition}

\begin{definition}[Extension on \(F_{GT}\) (Refined)]\label{def:fgt-ext}

	The refined version of \Cref{def:fgt-ext-unrefined} is shown in
	\figref{fig:ty-subst-refined}.

	\begin{figure*}[t]

		\normalsize
		\hrulefill{}
		\vspace{0.2cm}

		\begin{prooftree}
			\AxiomC{\(
			\forall \overrightarrow{G_i}. \left(
			\left(
			\bigwedge\limits^{i} \left(
			(\emptyset, P) \vdash G_i : {\tau_i}^{n_i} (\links{X_i})
			\right)
			\right)
			\Rightarrow (\emptyset, P) \vdash
			T
			{\overrightarrow{\left[G_i/x_i [\links{X_i}]\right]}}^{i}
			: {\tau}^{n}\,(\links{X})
			\right)
			\)}
			\RightLabel{Ty-Subst (rank \(n\))}
			\UnaryInfC{\(
			\left(
			{\overrightarrow{x_i [\links{X_i}] : {\tau_i}^{n_i} (\links{X_i})}}^{i}
			, P
			\right)
			\vdash T : {\tau}^{n}\,(\links{X})\)}
		\end{prooftree}

		\hspace{2.8cm}
		where
		\(n = \max n_i + 1\),
		\({\overrightarrow{x_i[\links{X_i}]}}^i = \mathit{ff} (T)\).

		\vspace{1em}
		\hrulefill{}
		\caption{Extension on \(F_{GT}\) (refined)}\label{fig:ty-subst-refined}
	\end{figure*}

\end{definition}

\begin{prop}
	The typing rules are well-defined
	even if we add \Cref{def:fgt-ext}.
\end{prop}

\begin{proof}


	In \Cref{def:fgt-ext},
	the ranks of the type on the left-hand side of \(\Rightarrow\), \(n_i\),
	are always smaller than the rank of the type \({\tau}^{n}\,(\links{X})\) on the consequent of the rule, \(n\).
	The typing rules in \figref{fig:fgt-rules-ranked} are defined
	so that the
	ranks do not increase when we read the rules upwards.
	Thus, the typing relation used for the left-hand side of
	\(\Rightarrow\) in the antecedent of Ty-Subst (of rank $n$)
	can be established using
	Ty-Subst with smaller ranks \(m \: (m < n)\) only
	(which may actually be used when, for example, the $G_i$'s
	contain abstraction atoms).
	Suppose all
	typing relations involving smaller ranks are well-defined.
	Then, since the typing relation we are about to define does not appear in the left-hand side of \(\Rightarrow\),
	we can ensure
	the well-definedness of the typing relation involving ranks up
	to $n$.
	Because the typing relation containing types with rank 0 only
	does not involve Ty-Subst and is therefore well-defined,
	by mathematical induction on rank, we can define a typing
	relation for all ranks.
\end{proof}

The existence of the typing rule defined in \Cref{def:fgt-ext} does not violate the soundness
since using the rule ensures that a program can be typed without such a rule
for all the possible graphs bound to graph contexts.

Let us consider the typing of a function that concatenates difference lists.
From now on, we omit the ``\((\emptyset, P) \vdash\)'' for brevity.
%
Suppose we have already proven the following (we will prove this in \Cref{sec:dlist-concat-proof}).
\begin{equation}\label{eq:1}
	\begin{array}{lll}
		\forall G_1, G_2. \Bigl(
		G_1 : {\mathit{nodes}^{n}} (Y, X) \land
		G_2 : {\mathit{nodes}^{m}} (Y, X) \\
		\Rightarrow
		\nu Z. (x [Z, X], y[Y, Z])
		\bigl[ G_1 \big/ x [Y, X] \bigr]
		\bigl[ G_2 \big/ x [Y, X] \bigr]  \\
		\quad \:\: : {\mathit{nodes}^{\max (n, m)}} (Y, X)
		\Bigr)
	\end{array}
\end{equation}

Then, we can type the function using
Ty-Sub,
Ty-Subst (rank \(\max (n, m) + 1\)),
and Ty-Arrow
as shown in \figref{fig:dlist-typing2}.

\begin{figure*}[t]

	\normalsize
	\hrulefill{}
	\vspace{0.2cm}

	\begin{prooftree}
		\AxiomC{\cref{eq:1}}
		\RightLabel{Ty-Sub}
		\UnaryInfC{\(
			\begin{array}{lll}
				\forall G_1, G_2. \Bigl(
				G_1 : {\mathit{nodes}^{n}} (Y, X) \land
				G_2 : {\mathit{nodes}^{m}} (Y, X) \\
				\Rightarrow
				\nu Z. (x [Z, X], y[Y, Z])
				\bigl[ G_1 \big/ x [Y, X] \bigr]
				\bigl[ G_2 \big/ x [Y, X] \bigr]
				: {\mathit{nodes}^{\max (n, m) + 1}} (Y, X)
				\Bigr)
			\end{array}
			\)}

		\RightLabel{Ty-Subst (rank \(\max (n, m) + 1\))}
		\UnaryInfC{\(\begin{array}{l}
				\left(
				\left(
				x [Y, X] : {\mathit{nodes}^{n}} (Y, X),
				y [Y, X] : {\mathit{nodes}^{m}} (Y, X)
				\right), P
				\right) \\
				\vdash \quad
				\nu Z. (x [Z, X], y[Y, Z])
				: {\mathit{nodes}^{\max (n, m) + 1}} (Y, X)
			\end{array}\)}
		\RightLabel{Ty-Arrow}
		\UnaryInfC{\((x [Y, X] : {\mathit{nodes}^{n}} (Y, X), P) \vdash
		\begin{array}{lcl}
			(\lambda\, y[Y, X] : {\mathit{nodes}^{m}} (Y, X). &   & ({\mathit{nodes}^{m}} (Y, X) \rightarrow         \\
			\qquad \nu Z. (x [Z, X], y[Y, Z])                 & : & \qquad {\mathit{nodes}^{\max (n, m) + 1}} (Y, X) \\
			)(Z)                                              &   & )\mbox{}^{\max (n, m) + 1} (Z)                   \\
		\end{array}
		\)}
		\RightLabel{Ty-Arrow}
		\UnaryInfC{\((\emptyset, P) \vdash
			\begin{array}{lcl}
				(\lambda\, x[Y, X] : {\mathit{nodes}^{n}} (Y, X).       &   & ({\mathit{nodes}^{n}} (Y, X) \rightarrow         \\
				\quad (\lambda\, y[Y, X] : {\mathit{nodes}^{m}} (Y, X). &   & \quad ({\mathit{nodes}^{m}} (Y, X) \rightarrow   \\
				\qquad \nu Z. (x [Z, X], y[Y, Z])                       & : & \qquad {\mathit{nodes}^{\max (n, m) + 1}} (Y, X) \\
				\quad )(Z)                                              &   & \quad )\mbox{}^{\max (n, m) + 1} (Z)             \\
				)(Z)                                                    &   & )\mbox{}^{\max (n, m) + 1} (Z)
			\end{array}
			\)}
	\end{prooftree}

	\vspace{1em}
	\hrulefill{}
	\caption{Type checking a function that concatenate difference lists}\label{fig:dlist-typing2}
\end{figure*}

Since the type system is monomorphic,
we cannot type the following program.

\[\begin{array}{l}
		((
		\begin{array}[t]{l@{\:\:}c@{\:\:}l}
			f       & : & ({\mathit{nodes}^{n}} (Y, X) \rightarrow         \\
			        &   & \quad ({\mathit{nodes}^{m}} (Y, X) \rightarrow   \\
			        &   & \qquad {\mathit{nodes}^{\max (n, m) + 1}} (Y, X) \\
			        &   & \quad )\mbox{}^{\max (n, m) + 1} (Z)             \\
			        &   & )\mbox{}^{\max (n, m) + 1} (Z)                   \\[1mm]
			x[Y, X] & : & {\mathit{nodes}^{l}} (Y, X)
		\end{array}, \\
		\:\:), P)                                                             \\
		\vdash \quad
		f \:\:  x[Y, X] \:\: (f \:\: x[Y, X] \:\: x[Y, X])
		: {\tau^{k}} (Y, X)
	\end{array}\]
We need to satisfy
\(n = \max (n, m) + 1\),
which is unsatisfiable.

Such programs can be typed introducing polymorphism for ranks.
However, this paper does not go into this and leaves it as future work.


}

\subsection{Proving the antecedent of the rule}\label{sec:dlist-concat-proof}

In order to apply \Cref{def:fgt-ext} to the present example,
we need to prove that, for any graphs to which
\(x [Y, X]\) and \(y [Y, X]\) can be mapped,
the substituted result
must have the type \({\mathit{nodes}}^{n} (Y, X)\), that is,
\[\begin{array}{l}
		\forall G_1, G_2. ((
		G_1 : {\mathit{nodes}}^{n} (Y, X) \land
		G_2 : {\mathit{nodes}}^{m} (Y, X)) \\
		\Rightarrow
		\nu Z. (x [Z, X], y [Y, Z])
		[G_1/x [Y, X]] [G_2/y [Y, X]]      \\

		\hspace{4.6em}= \nu Z.(G_1\angled{Z/Y}, G_2\angled{Z/X})
		: {\mathit{nodes}}^{\max (n, m)} (Y, X)).
	\end{array}\]

The above can be rewritten using Ty-Alpha as follows.

\begin{equation}\label{eq:3}
	\begin{array}{l}
		\forall G_1, G_2. (
		G_1 : {\mathit{nodes}}^{n} (Z, X) \land
		G_2 : {\mathit{nodes}}^{m} (Y, Z) \\
		\qquad \Rightarrow
		\nu Z.(G_1, G_2)
		: {\mathit{nodes}}^{\max (n, m)} (Y, X))
	\end{array}
\end{equation}

We prove this by induction on the derivation of the antecedents.
To do this, we need a lemma and a theorem.

\begin{lemma}\label{lem:prod-needed}
	For \(G: {\alpha}^{n} (\links{X})\), the rule
	Ty-Prod with \(\alpha/\norm{\links{X}}\) on its LHS is
	used in the derivation.
	Furthermore, only Ty-Cong and Ty-Alpha are used after the last
	application of Ty-Prod.
\end{lemma}

\begin{proof}

	Suppose we build a proof tree of \(G : {\alpha}^{n} (\links{X})\)
	bottom-up.
	Since \(G\) is a value,
	we can only use Ty-Cong, Ty-Alpha,
	{Ty-Subst, Ty-Sub}
	and Ty-Prod
	until a \(\lambda\)-abstraction atom appears.
	Ty-Alpha,
	{Ty-Subst, Ty-Sub,}
	and Ty-Cong only inherit the annotated type
	from the antecedent
		{
			(although they may changes the rank)
		} 
	so they alone cannot make the annotated type a type variable atom.
	If Ty-Prod does not appear
	but a \(\lambda\)-abstraction atom appears and Ty-Arrow is used,
	the annotated type becomes an arrow and not a type variable.
	Therefore, there must exist a Ty-Prod whose annotated type
	has the functor \(\alpha/\norm{\links{X}}\).
\end{proof}

\begin{theorem}\label{th:graph-decomp}
	For \(G: {\alpha}^{n} (\links{Y})\),
	if the production rule used by last Ty-Prod was
	\(\alpha (\links{X}) \longrightarrow \Tauu\),
	there exists
	\({\overrightarrow{G_j}}^{j}\)
	such that
	\(G \equiv\Tauu' [{\overrightarrow{G_j / \tau_j\, (\links{X_j})}}^{j}]\)
	where
	\begin{itemize}
		\item
		      \(\Tauu' = \Tauu {\overrightarrow{\angled{Y_i / X_i}}}^{i}\),

		\item
		      \(\tau_j\, (\links{X_j})\) are
		      all the type atoms
		      appearing in \(\Tauu'\),

		\item
		      \({\overrightarrow{G_j : {\tau_j}^{n_j} (\links{X_j})}}^{j}\),
		      and

		\item
		      \(\max n_j = n\).
	\end{itemize}
\end{theorem}

\begin{proof}
	By induction on the derivation of \(G: {\alpha}^{n} (\links{Y})\)
	after the last application of Ty-Prod using \Cref{lem:prod-needed}.
\end{proof}

%
%

Consider the case if the rule Ty-Prod used last
{
in the derivation of \(G_1: {\mathit{nodes}}^{n} (Z, X)\)
}
was the one with the following production rule.
\[\begin{array}{l}
		{\mathit{nodes}} (X_2, X_1) \\
		\longrightarrow
		\nu X_3 X_4.(\mathrm{Cons} (X_3, X_4, X_1), \mathit{nat}\,(X_3), \mathit{nodes}\,(X_2, X_4)),
	\end{array}\]
Then,
by \Cref{th:graph-decomp},
we can decompose the graph
\(G_1 \equiv
\nu X_3 X_4.(\mathrm{Cons} (X_3, X_4, X), G_3, G_4)\)
into \(G_3\) and \(G_4\),
where
\(G_3: {\mathit{nat}}^{0} (X_3)\)
and
\(G_4: {\mathit{nodes}}^{n} (Z, X_4)\).
And we can proceed verification by checking if the target graph
has the desirable type for all possible values of
\(G_3\) and \(G_4\).

We prove \cref{eq:3} by induction on the derivation of
\(G_1: {\mathit{nodes}}^{n} (Z, X)\).
We split the cases based on the rule Ty-Prod used last in the derivation
and decompose the graph using \Cref{th:graph-decomp}.


For brevity,
we denote the graph
\(G\) of the type \({\alpha}^{n} (\links{X})\)
as
\({\underline{\alpha}}^{n} (\links{X})\)
and omit \(\forall G\).
Then \cref{eq:3} can be rewritten as
\[\begin{array}{l}
		\nu Z. (\ulrm{nodes}{n}_1 (Z, X), \ulrm{nodes}{m}_2 (Y, Z))
		: {\mathit{nodes}}^{\max (n, m)} (Y, X).
	\end{array}\]

The inference rule
(or rule \emph{scheme}, precisely speaking)
that splits the cases
by the last application of Ty-Prod to derive
\({\underline{\beta}}^{n}_j\)
in \(G\)
is expressed in the following form.
Here,
\(G_i\)
is the graph such that the last production rule used in the
derivation of
\({\underline{\beta}}^{n}_j\)
is
\(P_i\).

\begin{prooftree}
	\AxiomC{\(G_1 : {\alpha}^{l} (\links{X})\)}
	\AxiomC{\hspace{-12pt}\(\dots\)\hspace{-12pt}}
	\AxiomC{\(G_m : {\alpha}^{l}  (\links{X})\)}
	\RightLabel{Case \({\underline{\beta}}^{n}_j\)}
	\TrinaryInfC{\(G: {\alpha}^{l} (\links{X})\)}
\end{prooftree}


The concatenation of difference lists can be
verified as shown in \figref{fig:dlist-concat-proof},
where the arrow \(\hookleftarrow\) refers to using the induction hypothesis.

\begin{figure*}[t]
	\small
	\footnotesize
	\begin{prooftree}
		\AxiomC{\(\ulrm{nodes}{m}_2 (Y, X) : {\mathit{nodes}}^{m} (Y, X)\)}
		\RightLabel{Ty-Cong}
		\UnaryInfC{\(
		\nu Z. (X \bowtie Z, \ulrm{nodes}{m}_2 (Y, Z)) : {\mathit{nodes}}^{m} (Y, X)
		\)}
		\RightLabel{Ty-Sub}
		\UnaryInfC{\(
		\nu Z. (X \bowtie Z, \ulrm{nodes}{m}_2 (Y, Z)) : {\mathit{nodes}}^{\max (n, m)} (Y, X)
		\)}

		\AxiomC{\(\ulrm{nat}{0}_3 (W_1) : {\mathit{nat}}^{0}(W_1)\)}

		\AxiomC{\(\nu Z.(\ulrm{nodes}{n}_4 (Z, X), \ulrm{nodes}{m}_2 (Y, Z)):
		{\mathit{nodes}}^{\max (n, m)} (Y, X)\)}
		\RightLabel{Ty-Alpha}
		\UnaryInfC{\(\nu Z.(\ulrm{nodes}{n}_4 (Z, W), \ulrm{nodes}{m}_2 (Y, Z)):
		{\mathit{nodes}}^{\max (n, m)} (Y, W)\)}
		\RightLabel{Ty-Prod \(P_2\)}
		\BinaryInfC{\(\nu W. (\mathrm{Cons} (\ulrm{nat}{0}_3, W, X), \nu Z.(\ulrm{nodes}{n}_4 (Z, W), \ulrm{nodes}{m}_2 (Y, Z))):
		{\mathit{nodes}}^{\max (n, m)} (Y, X)\)}
		\RightLabel{Ty-Cong}
		\UnaryInfC{\(\nu Z. (\mathrm{Cons} (\ulrm{nat}{0}_3, \ulrm{nodes}{n}_4 (Z), X), \ulrm{nodes}{m}_2 (Y, Z)):
		{\mathit{nodes}}^{\max (n, m)} (Y, X)\)}

		\RightLabel{Case \(\ulrm{nodes}{n}_1\)}
		\BinaryInfC{\(\nu Z. (\ulrm{nodes}{n}_1 (Z, X), \ulrm{nodes}{m}_2 (Y, Z)):
		{\mathit{nodes}}^{\max (n, m)} (Y, X)\)}
	\end{prooftree}
	\begin{tikzpicture}[overlay,remember picture]
		\begin{scope}[shift={(0.3, 0)}]
			\coordinate[] (D) at (15.9, 2.25);
			\coordinate[] (C) at (17.3, 2.25);
			\coordinate[] (B) at (17.3, 0.3);
			\coordinate[] (A) at (11,   0.3);
			\draw[<-,rounded corners=4pt] (A)--(B)--(C)--(D);
		\end{scope}
	\end{tikzpicture}

	\caption{Verifying concatenation of difference lists}\label{fig:dlist-concat-proof}
\end{figure*}

\section{Automatic verification on the extended type system}\label{sec:autoverify}

\renewcommand{\ulrm}[2]{{\underline{\mathit{#1}}}}

In Section~\ref{sec:fgt-ext},
we typed the program by manually applying structural induction to the target program.
In this section, we describe a method to do this automatically.
	{
		From now on, we handle the cases where ranks are all zero and omit them.
		Extending the algorithm to handle general rank is future work.

	}

We construct a proof tree like what we have shown in
\figref{fig:dlist-concat-proof}
bottom-up.
Given \(G: \alpha (\links{X})\),
we can use the following strategies to verify those programs.
\begin{description}
	\item[Ty-Prod:]
		If we get a constructor atom \(C/n\) from \(G\),
		we can check whether an annotated type name atom \(\alpha (\links{X})\)
		can be derived from the target graph
		using Prod with a production rule with
		\(\alpha (\links{X})\) on the LHS
		and \(C/n\) on the RHS\@.
		However, in order to use a production rule,
		the subgraphs in \(G\) must have the types necessary for the derivation.
		For this reason, the type checker is performed inductively on the subgraphs.

	\item[Case \(\underline{\beta}_i\,\):]
		If we get a
		type annotated graph \(\underline{\beta}_i\)
		from \(G\),
		we decompose it using \Cref{th:graph-decomp}.
		Then check if \(G\) with its subgraph
		\(\underline{\beta}_i\) thus decomposed
		has type
		\(\alpha (\links{X})\).

	\item[\(\hookleftarrow\,\):]
		Induction hypotheses are used when applicable.

\end{description}


However, it is not that easy to do this automatically.
Especially for more complex examples.
\begin{enumerate}
	\item
	      We cannot easily separate a graph into subgraphs when
	      using a production rule.
	      It is difficult to automatically separate and guess the type of a subgraph,
	      prove it as a subproblem, and proceed with the proof using it
	      without any prior preparation.

	\item
	      The possibility that links may be \emph{fused} later
	      makes it difficult to get the correspondence of link names in the target graph and the applying production rule.

	      Remember the production rules for leaf-linked trees in \Cref{ex:prod-lltree}.
	      Here, we want to type check the following graph.
	      \[\begin{array}{lcl}
			      \nu Y. \mathrm{Node} (L, \mathrm{lltree} (Y, R), X), \mathrm{Leaf} (\mathrm{nat}, L, Y)
			      : \mathrm{lltree} (L, R, X)
		      \end{array}\]
	      In this example, we try to apply the second rule
	      \[\begin{array}{l@{~~}c@{~~}l}
			      \mathrm{lltree} (L, R, X) & \longrightarrow & \nu Y. \mathrm{Node} (\mathrm{lltree} (L, Y), \mathrm{lltree} (Y, R), X).
		      \end{array}\]
	      In this rule, the first link of the atom
	      \(\mathrm{Node}/3\) is a(n anonymous) local link, say \(Y_1\),
	      but the corresponding link in the target graph is the free link \(L\).
	      Therefore, it is necessary to proceed with the information that \(Y_1\) will be fused to \(L\) later,
	      and to check that the fusion occurs before the local link \(Y_1\) leaves the scope.
	      Implementing this becomes complex with more similar
	      cases and is not that easy.
	      In addition, it is not trivial to add the structural
	      induction hypothesis in this process
	      and apply it.

	\item
	      A strategy is also needed for the decomposition of
	      annotated contexts.
	      In \figref{fig:dlist-concat-proof},
	      we decomposed \(\ulrm{nodes}{n}_1 (Z, X)\).
	      If we decomposed \(\ulrm{nodes}{n}_2 (Y, Z)\),
	      we would not get the form to which induction hypothesis
	      can be applied
	      and verification would fail.
\end{enumerate}

Therefore, we restrict the production rules
to facilitate disassembly into subgraphs by
introducing the notion of a \emph{root link}.
Also, fusions are absorbed first to prevent link fusion from occurring later.
And we decompose annotated contexts from the one holding a free root
link (e.g., \(\ulrm{nodes}{n}_1\) holding \(X\) in the proof goal of
\figref{fig:dlist-concat-proof}).

\subsection{Constraints on production rules}\label{sub:constraints}

The type system \(F_{GT}\) defined so far has
imposed no restriction on production rules,
even disconnected graphs (multisets) could be handled.
However,
here, we design the type system to efficiently support data structures of practical importance.

In order to handle graphs inductively with production rules easier,
we introduce the notion of root links.

\begin{definition}[Root link]

	We call the last link of each atom as its \emph{root link}.

\end{definition}

We give a restriction on production rules so that
we can find a spanning tree of a graph
by traversing the root links.
Since a spanning tree can be found for any connected graph,
we can arrange the ordering of links of individual atoms in such a way
that the root links form the edges of a spanning tree.
Thus the restriction on production rules
will not essentially sacrifice the expressive power of the data
structure for practical programs.
We call a link \(X\) the \emph{root link of a graph} \(G\)
if every atom in the graph can be reached through their root links
from \(X\).

\begin{definition}[Constraints on production]\label{subsub:production-rules}

	A production rule should have the form
	\(\alpha (\links{X}, R) \longrightarrow \tau\),
	where the \(\tau\) should be one of the following.
	\begin{enumerate}
		\item
		      one or more fusions.

		\item
		      has one constructor atom \(C (\links{Y}, R)\),
		      zero or more type variable atoms \(\alpha_i (\links{Y_i})\),
		      zero or more fusions, and
		      zero or more arrow atoms
		      and satisfies all the following conditions.
		      \begin{enumerate}
			      \item
			            The root link \(R\) of \(C (\links{Y}, R)\) occurs free
			            in \(\tau\).

			      \item
			            The root link \(R_i\) of
			            a type variable atom
			            \(\alpha_i (\dots, R_i)\)
			            should satisfy
			            \(R_i \in \{\links{Y}\}\)
			            and
			            all the \(R_i\)'s are mutually distinct.
		      \end{enumerate}
	\end{enumerate}
\end{definition}

All the examples we have introduced in \Cref{sub:graph-types} satisfy these constraints.
Therefore, we claim that most of the practical examples are covered
even with the restrictions.

\subsection{Fusion elimination}\label{subsec:normal-form-trans}

Since fusion (\(\bowtie\)) is difficult to handle,
we attempt to eliminate fusions (\(\bowtie\))
except when they are generated directly from the annotated type variable atom
by
merging of production rules.

\begin{definition}[Fusion elimination]

	Let \(P_{\bowtie}\) denote the set of production rules that
	include fusion.
	And let \(\overline{P_{\bowtie}}\) denote the set of
	production rules without fusion.
	%
	%
	For each production rule in \(\overline{P_{\bowtie}}\),
	we apply the production rules in
	\(P_{\bowtie}\) to the
	type variable atom in the RHS of the rule.
	This is done in \(n^2\) ways for \(n\) type variable atoms
	to cover all combinations.
	We add the newly created rules,
	which includes the original one,
	to \(P'\).
	%
	We also add the rules that have no type variable atoms on RHS to \(P'\).
	%
	If there exist rules in \(P'\) and \(P_{\bowtie}\)
	which have the annotated type variable \(\alpha
	(\links{X})\) on the LHS,
	we add the rule whose LHS are replaced with \(\alpha_{\bowtie} (\links{X})\) to \(P'\).

	Finally, we replace the annotated type variable \(\alpha (\links{X})\) with
	\(\alpha_{\bowtie} (\links{X})\).

\end{definition}

%
We have observed that it is not always possible to eliminate fusion in
this way.
However, all of our practical examples can be successfully transformed by this method.
A more refined method of fusion elimination
and a rigorous proof that the production rules obtained by this operation
are equivalent to the original ones will be the subject of future work.

If fusion elimination succeeds,
we can say that fusion will not appear ``later''
when the production rule is applied backwards (Ty-Prod).
%
On the other hand,
we cannot deny the possibility of occurrence of unabsorbable fusion
when applying production rules to decompose graphs (Case).
However, this did not happen in our examples.

Once we eliminate fusions,
it will be easy to check the correspondence of links.
Firstly, we \(\alpha\)-convert link names so that all the link names
are distinct.
Then, the correspondence of links in the target graph and the annotated type can be checked as follows.
If they are free links, check if they have the same name.
If the links are local links,
we check the correspondence between the link in the target graph and the link in the annotated type
based on mapping.
If the correspondence has not yet been established,
add a new correspondence.
If the correspondence is already in place,
we check that it is satisfied.
\Figref{table:check-link-name} shows the algorithm to check the correspondence of links.

\begin{figure}[t]
	\removelatexerror{} 
	\begin{algorithm}[H]
		\DontPrintSemicolon{}
		\Let{\(\mathit{check\_link\_name}\)\\
			\Indp{}
			\(L\)          \tcp*{A set of local links of the target graph}
			\(f\)                  \tcp*{A mapping from the links in annotation to the links in the target graph}%
			\((\)
			\(X\),                 \tcp*{The link in the target graph}
			\hspace{0.18cm}\(Y\)   \tcp*{The link in the annotation}
			\()\)\\
			\Indm{}
		}{%
			\eIf{\(X \notin L\)}{%
				\leIf{\(X \notin \mathrm{dom} (f) \land X = Y\)}{\(\mathrm{Some}\ f\)}{\(\mathrm{None}\)}
			}{%
				\lIf{\(Y \mapsto \mathrm{None} \in f\)}{%
					\(\mathrm{Some}\ (f\ \textit{updated with}\ Y \mapsto \mathrm{Some}\ X)\)}
				\lElseIf{\(Y \mapsto \mathrm{Some}\ X \in f\)}{\(\mathrm{Some}\ f\)}
				\lElse{\(\mathrm{None}\)}
			}
		}
	\end{algorithm}
	\caption{Check link name}\label{table:check-link-name}
\end{figure}

\subsection{The algorithm}\label{sub:graph-typechecking}

It will be a little troublesome to implement the backward application of a production rule
to handle the reverse execution of Ty-Prod.
Thus, we will first apply the production rule to the annotated type
and then remove the constructor atom both on the target graph and the annotated type.
Note that this will result in allowing graphs
in the annotation during the execution of this algorithm,
which we refer to as an \emph{annotated graph}.

\Figref{table:typechecker} shows the outline of the algorithm.
The function \(\mathit{check} (G, \alpha(\links{X}, R), P)\) checks that
\((\emptyset, P) \vdash G: \alpha(\links{X}, R)\)
where \(G\) possibly includes \(\underline{\beta} (\links{Y})\);
type annotated graph \(G_{\beta}\) where \(G_{\beta} : \beta (\links{Y})\).
The algorithm runs recursively with \(\mathit{helper}\) function (line \(6\))
on the atoms/type annotated graph
with a root link \(R\) of the target graph \(G\) and the annotated graph \(\Tauu\).

Line \(12\) checks that the graph \(G\) has type \(\Tauu\) \emph{trivially}.
For example, \(G\) maybe the type annotated graph whose annotated type was \(\Tauu\)
or a \(\lambda\)-abstraction atom, whose typing relation can be checked as
the same as the other functional language
(except that we may need to apply this algorithm recursively for the graphs in its body expression).

From line \(13\), we split the cases by the atom
with the root link of the target graph and the annotated graph.
If both atoms have constructor names with the same functor, then we remove the atoms
and run the algorithm recursively to all the subgraphs traversable from their arguments.

If the atom in the annotated graph is a type variable atom \(\alpha (\links{Y})\),
then we first try to use induction hypotheses \(H\) (line \(23\) and line \(28\)).
Notice that we can use congruence rules (Ty-Cong) and \(\alpha\)-conversion of free links (Ty-Alpha)
to absorb the syntactic difference
between \((G: \Tauu)\) and hypothesis in \(H\).

If we cannot prove it by the hypothesis,
then we should proceed with the construction of the proof tree with
Ty-Prod or Case.
If the root of the target graph is a constructor atom \(C_G (\links{X})\) (line \(22\)),
then we apply the production rules
whose LHS is \(\alpha/\norm{\links{Y}}\)
and check there \emph{exists} a way to successfully construct a sub-proof.
Notice that we add the current typing relation to the induction hypotheses.
If the root of the target graph is a type annotated graph \(\underline{\beta} (\links{X})\) (line \(27\)),
then we decompose the graph using the production rules of last Ty-Prod
and check \emph{all} of them satisfies the type.

Although we did not mention it in our pseudocode
but we need to make sure that the links \(\links{X}\) and \(\links{Y}\) have a proper correspondence
using the function
we have shown in \figref{table:check-link-name}.

\begin{figure}[t]
	\removelatexerror{} 
	\begin{algorithm}[H]
		\DontPrintSemicolon{}
		\Let{\(\mathit{check}\) \((\){\\
				\Indp{}
				\(G\),                         \tcp*{Target graph}%
				\(\alpha (\links{X}, R)\), \tcp*{Annotated type atom}
				\(P\)                          \tcp*{Production rules}
				\Indm{}
			}\()\)}{%
		\LetRec{\(\mathit{helper}\) \((\){\\
				\Indp{}
				\(R\),                     \tcp*{Root link}%
				\(G\),                     \tcp*{Target subgraph}%
				\(\Tauu\),                 \tcp*{Annotated graph}%
				\(H\)                      \tcp*{Induction hypotheses}
				\Indm{}
			}\()\)}{%
		\lIf{trivially\ \(G: \Tauu\)}{\(\mathrm{true}\)}
		\Switch{%
		\(
		(
		v\, (\links{X},R) \textrm{~or~} \alpha (\links{X},R) \textrm{~in~} G
		,\;
		\tau\, (\links{X},R) \textrm{~in~} \Tauu
		)\)
		}{%
		\Case{\(C_G (\links{X}, R), C_\Tauu (\links{Y}, R)\)}{%
			\lIf{\(
				C_G/\norm{\links{X}}
				\neq C_\Tauu/\norm{\links{Y}}
				\)}{\(\mathrm{false}\)}
			\Else{%
				\(\forall i\).\\
				\If{\(X_i\) and \(Y_i\) are the roots of the non-empty subgraph \(G_i\) and \(\Tauu_i\)
				}{\(\mathit{helper}\ (R, G_i, \Tauu_i, H)\)}
				\Else{%
					\(X_i\) and \(Y_i\) are not the root of atoms in \(G\) and \(\Tauu\)
				}
			}
		}
		\Case{\(C_G (\links{X}), \alpha (\links{Y})\)}{%
		\((G: \Tauu) \in H\ \lor\)\\
		\(\exists (\beta (\links{Z}) \longrightarrow \Tauu') \in P\)
		such that\\
		\(\alpha/\norm{\links{Y}} = \beta/\norm{\links{Z}}\ \land\)\\
		\(\mathit{helper}\ (R, G, \Tauu'{{\overrightarrow{\angled{Y_i/Z_i}}}^{i}}, \{G: \Tauu\} \cup H)\)
		}
		\Case{\(\underline{\beta} (\links{X}), \alpha (\links{Y})\)}{%
			\((G: \Tauu) \in H\ \lor\)\\
			\(\forall (r \textit{ with } \beta/\norm{\links{X}} \textit{ on LHS } \in P)\).\\
			\(\mathit{helper}\
			(R, G \textit{ decomposed } \underline{\beta} (\links{X}) \textit{ with } r, \Tauu, H)\)
		}
		\lCase{\upshape {\bfseries otherwise}}{\(\mathrm{false}\)}
		}
		}{%
		\(
		\mathit{helper}\
		(R, G, \alpha (\links{X}, R), \emptyset{})\)
		}
		}
	\end{algorithm}
	\caption{Graph type checker}\label{table:typechecker}
\end{figure}

\begin{theorem}
	The algorithm in \figref{table:typechecker} is sound.
\end{theorem}
\begin{proof}
	This is straightforward since we are constructing a proof tree.
	There is a concern that soundness may be violated when the induction hypothesis is used,
	but this is not a problem.
	This is because the size of the graph gets strictly smaller
	when the type checker applies the structural induction.
	The structural induction hypothesis is added on line 26,
	where a production rule is applied to the annotation,
	and the root of the annotated graph always becomes a constructor atom.
	Therefore, the type checker does not proceed to the cases except in line 14
	in the recursion,
	and if this branch succeeds,
	the constructor atom is removed,
	reducing the size of the graph.
	Therefore, it is sound by the infinite descent method.
\end{proof}


\section{Related work}\label{sec:related-work}

Since graphs and its operations are more complex than trees,
there are diverse formalisms for graphs and graph types.

\subsection{Typing frameworks for graphs}
Structured Gamma~\cite{structuredgamma}
is a typing framework for graphs,
in which types are defined by production rules in context-free graph grammar.
%
Shape Types~\cite{shapetypes}
are similar but the following restrictions are imposed
on type definitions to ensure completeness of type checking:
(i) the state space of type checking must be confluent,
and (ii) graphs supplemented during the type checking
must consist only of a finite number of symbols.
%
With context-free graph grammar,
we can express a broad and expressive class of types.
However, type checking becomes harder
and hence it does not cover some practical operations.
For example, the concatenation of difference lists
and the pop operation from the tail of them cannot be checked
by either Shape Types or Structured Gamma.
%
In this research, we restrict the target grammar
so that we can verify practical operations by structural induction.

With Graph Types~\cite{graphtypes}, we can define types of algebraic data structures accompanied by \emph{extra edges},
where the destination of an extra edge is specified by a \emph{routing expression}.
A routing expression is a regular expression over small-step traverse operations,
which describes the relative position of the destination of an extra edge,
and the actual destination can be automatically computed based on it.
In addition, Graph Types provide a decidable monadic second-order logic on the types as a way of formal verification and automatic program generation.
For example, a constant-time concatenation of doubly-linked lists as modification of pointers can be deduced by the logic.

Our type system \(F_{GT}\) and Graph Types share the ideas that typed graphs consist of a canonical spanning tree and auxiliary edges, and types are defined by production rules.
On the other hand, auxiliary edges and their modification are \emph{computed} based on routing expressions in Graph Types, whereas they are described by users and \emph{verified} by the types in our method.
In addition, pattern matching based on the types can be described in our language \(\lambda_{GT}\).

\subsection{Functional language with graphs}
FUnCAL~\cite{funcal} is a functional language that supports graphs as a first-class data structure.
This language is based on an existing graph rewriting language, UnCAL~.
In UnCAL (and FUnCAL), graphs may include back edges
and their equality is defined based on bisimulation.
FUnCAL comes with its type system but does not support pattern matching
for user-defined data types, which classic functional languages support for ADTs.

Functional programming with structured graphs~\cite{structuredgraphs}
can express recursive graphs using recursive functions, i.e., \textbf{let rec} statements.
Since they employ ADTs as the basic structure,
they can enjoy type-based analysis based on the traditional type system.
On the other hand, we can do further detailed type analysis
by our language and type system.

Initial algebra semantics for cyclic sharing tree structures~\cite{hamana2010}
discusses how to express graphs by \(\lambda\)-expressions.
However, there is a large gap between \(\lambda\)-expressions and pointer structures.
On the other hand, we defined a graph based on nodes and hyperedges,
which has a clear correspondence to a pointer structure.
This style is rather suitable for future implementation.
In addition, they do not support user-defined graph types
or verification based on them.

\subsection{Separation Logic}

Our approach is in contrast with the analysis of pointer manipulation programs
using Separation Logic \cite{separation-logic}, shape analysis \cite{shape-analysis}, etc.

%
Firstly, the target languages differ in many ways.
Separation Logic and shape analysis normally handle low-level
imperative programs using heaps and pointers.
In contrast,
we dispense with destructive operations and adopt
pattern matching over graphs provided by
the new higher-level language \(\lambda_{GT}\), which
abstracts address, pointers and heaps away, and features hyperlinks and
operations on them including fusion and hiding.

Secondly, we pursue a lightweight,
automatic type system for functional languages
rather than Hoare-style general verification for imperative languages.
Separation Logic allows us to use \textit{pure formulae}
that represent various non-spatial properties.
The only thing that seems to correspond to pure formulae
in our type system is fusion (which can be regarded as \(x =
y\) in Separation Logic).
This design choice reflects the fact that our goal is not
a formal system for software verification but a programming
language and its type system.
%

The problem discussed in \Cref{sec:fgt-ext}, verification of
an inductively defined structure with structural induction,
is close to the entailment problem of
inductive predicates with symbolic heaps in Separation Logic,
sometimes referred to as
\emph{SLRD} (Separation Logic with Recursive Definitions).
Cyclist~\cite{BrotherstonGP12}
performs automatic verification of the problem.
However, the algorithm requires dynamic checking of the soundness condition.
On the other hand, we have restricted graph grammar and
proved the soundness statically as a (meta-){}theorem.
Antonopoulos et al.\ \cite{general-ind-pred2014} show that the
entailment problem of general SLRD is undecidable.
Therefore, decision procedures for them impose some restrictions on SLRD.
Iosif et al.\ \cite{SLRDbtw2013} propose a sub-class of SLRD,
SLRD\(_{\textit{btw}}\),
which handles graphs with bounded treewidth.
The restrictions imposed on the recursive definitions are similar to the
restrictions we have introduced in \Cref{sub:constraints}.
However, they do not allow empty graphs
and cannot handle a difference list without elements.
Tatsuta et al.~\cite{tatsuta2019} has imposed further restriction
to SLRD\(_{\textit{btw}}\)
which corresponds to the notion of \emph{root link} in ours.
A precise comparison of the algorithms in \Cite{tatsuta2019} and our
technique
will be the subject of future work.


\section{Conclusions and further work}\label{sec:conclusion}

In this study, we proposed a new functional language \(\lambda_{GT}\) that handles graphs
as a first-class data structure
with declarative operations based on graph transformation.

First, we formalized the formal syntax and semantics of \(\lambda_{GT}\)
in a syntax-directed manner,
incorporating HyperLMNtal into a call-by-value \(\lambda\)-calculus.

Second, we developed a new type system \(F_{GT}\) that
empolys HyperLMNtal rules as production rules to deal with
data structures more complex than trees.

Third, we extended the type system to support more powerful verification such as
concatenation of difference lists.
Then we developed an algorithm to automatically verify programs
with the extended type system using structural induction.

Finally, we address future work that is not mentioned in previous sections.

\subsection{Extend the type system to handle untyped graph contexts}

In this paper, we introduced dynamic type checking (\Cref{sub:dynamic-checking})
and excluded untyped graph contexts.
However, verification with untyped graph contexts is necessary not just to
reduce the programmer's extra effort
since there exist programs that cannot be succinctly handled without untyped graph contexts.
For example, matching the leftmost leaf in a leaf-linked tree is possible
in \(\lambda_{GT}\) 
using a template consisting of the leftmost leaf and an untyped graph
context for the rest of the tree. 
However, we cannot denote the type of the untyped graph context using the type of the leaf-linked tree
because it is not a tree.

\subsection{Full implementation of the language and the type system}

We have implemented the type checker to verify operations over graphs.
However, 
implementation of the language with the full type system including
arrows is a future work.
We believe that it is straightforward to implement the type system.
However, implementation of the efficient runtime has
many things to be considered including deeper static analysis of programs
(such as the guarantee of immutability using ownership checking) 
to allow destructive operations on graphs 
without forcing imperative programming on users.

\subsection{Extension on the type system: polymorphism and type inference}

The proposed type system \(F_{GT}\) is monomorphic.
We can only define difference lists with a specific element type,
though introducing generic data types as in other functional languages
could be done in the same way.

However, for more complex data structures, introducing polymorphism
may be not that straightforward 
since we have introduced more powerful operations than the other languages
such as concatenation of difference lists.
In \(\lambda_{GT}\), concatenation of difference lists can be done without
explicitly handling constructor atoms,
which may be typeable as a generic function.
However, since operations on data structures may not result in
data structures of the same type,
we may need to verify programs with the type information of the inputs,
which seems to be a little incompatible with polymorphism.

The same thing can be said for type inference.
Since we allow powerful operations
over data structures
without explicitly denoting constructor names,
it may be more difficult than in other functional languages
and may require some non-obvious ingenious techniques.

\bibliographystyle{plain}

\appendix

In this appendix, we give proofs for the propositions and theorems that appeared in this paper.



\section{Proof of properties of HyperLMNtal}




%
%
%
%
%
%

\begin{lemma}[Elimination of \(\nu\) which bounds no link name]\label{lem:absorb-link-creation}
	\[\nu X.G \equiv G \mbox{ where } X \notin \mathit{fn}(G)\]
\end{lemma}

\begin{proof}{}{}
	\[\begin{array}{llr}
			                               & \nu X.G                                                   &
			\\
			\equiv_{\mbox{\scriptsize E5}}
			                               & \nu X.(\zero, G)
			                               & \because G \equiv_{\scriptsize \mbox{E1}} (\zero, G)
			\\
			\equiv_{\mbox{\scriptsize E10}}
			                               & (\nu X.\zero, G)
			                               & \because X \notin \mathit{fn}(G)
			\\
			\equiv_{\mbox{\scriptsize E4}}
			                               & (\zero, G)
			                               & \because \nu X.\zero \equiv_{\mbox{\scriptsize E8}} \zero
			\\
			\equiv_{\mbox{\scriptsize E1}} & G
		\end{array}\]
\end{proof}

\begin{lemma}[Elimination of a futile link substitution]\label{lem:futile-link-substitution}
	\[G\angled{X/X} = G\]
\end{lemma}

This is not as obvious as it may seem.
The reason is that it cannot be naively ruled out that a $\alpha$-conversion
may be performed during the hyperlink assignment,
resulting in a congruent but syntactically different graphs.

\begin{proof}{}{}
	We prove by inducition on graphs.
	It is trivial for
	$\zero$,
	$p(X_1, \ldots, X_m)$,
	$(G, Q)$,
	$(G \means Q)$.

	\vspace{0.5em}\noindent%
	Case \(\nu Y.G\):
	\begin{quote}
		\(
		(\nu Y.G)\angled{X/X} \defeq
		\left\{
		\begin{array}{ll}
			\nu Y.G             & \mbox{if } Y = X \vspace{1em} \\
			\nu Y.G\angled{X/X} & \mbox{if } Y \neq X           \\
			\multicolumn{2}{l}{ = \nu Y.G \because \mbox{induction hypothesis}}
		\end{array}
		\right.
		\)


		Since \(Y \neq X \land Y = X\) can never happen,
		there is no possibility of $\alpha$-conversion of links
		(which could have resulted in loss of syntactic equality) to avoid variable capture.
	\end{quote}

\end{proof}


\begin{proof}[Proof of \Cref{th:alpha-equiv}]


	We are using
	\Cref{lem:absorb-link-creation}
	and \Cref{lem:futile-link-substitution}.
	We consider the case where the free hyperlink to be substituted appears
	and the case where it does not.
	The latter case seems obvious,
	but it is not because of the possibility of $\alpha$-conversion due to hyperlink substitution.
	We prove the former first, and then transform the latter into a form that allows us to use the former.

	\vspace{0.5em}\noindent%
	Case \(X \in \mathit{fn}(G)\):
	\begin{quote}
		\(\begin{array}{ll}
			\nu X. \nu Y. (Y \bowtie X, (X \bowtie Y, G))
			\\
			\equivby{E5, E3}
			\nu X. \nu Y. ((Y \bowtie X, X \bowtie Y), G)
			\\
			\equivby{E5, E10}
			\nu X. (\nu Y. (Y \bowtie X, X \bowtie Y), G) \\
			\hspace{2em}\because
			Y \notin \mathit{fn}(G)
			\\
			\equivby{E5, E6}
			\nu X. (\nu Y. X \bowtie X, G)                \\
			\hspace{2em}\because
			(X \bowtie Y)\angled{X/Y} = X \bowtie X
			\\
			\equivby{E5, E10}
			\nu X. \nu Y. (X \bowtie X, G)                \\
			\hspace{2em}\because
			Y \notin \mathit{fn}(G)
			\\
			\equivby{E9}
			\nu Y. \nu X. (X \bowtie X, G)
			\\
			\equivby{E6, \Cref{lem:futile-link-substitution}}
			\nu Y. \nu X. G                               \\
			\hspace{2em}\because
			G\angled{X/X} = G
			\\
			\equivby{\Cref{lem:absorb-link-creation}}
			\nu X. G                                      \\
			\hspace{2em}\because
			Y \notin \mathit{fn}(G)
		\end{array}\)

		and

		\vskip.5\baselineskip
		\(\begin{array}{ll}
			\nu X. \nu Y. (Y \bowtie X, (X \bowtie Y, G)) &
			\\
			\equivby{E2, E3, E5, E9}
			\nu Y. \nu X. (X \bowtie Y, (Y \bowtie X, G)) &
			\\
			\equivby{E5, E6}
			\nu Y. \nu X. (Y \bowtie Y, G\angled{Y/X})    &
			\\
			\equivby{E5, \Cref{lem:absorb-link-creation}}
			\nu Y. (Y \bowtie Y, G\angled{Y/X})                                                       \\
			\hspace{2em}\because
			X \notin \mathit{fn}((Y \bowtie Y, G\angled{Y/X}))
			\\
			\equivby{E6}
			\nu Y.G\angled{Y/X}                                                                       \\
			\hspace{2em}\because
			\mbox{by \Cref{lem:futile-link-substitution}} (G\angled{Y/X})\angled{Y/Y} = G\angled{Y/X} \\
		\end{array}\)
		\vskip.5\baselineskip

		Thus,
		\(\nu X.G \equiv \nu Y.G\angled{Y / X}\)

	\end{quote}

	\vspace{0.5em}\noindent%
	Case \(X \notin \mathit{fn}(G)\) :
	\begin{quote}
		In this case, we use the previous proof by first adding a free hyperlink \(X\) using (E7).

		\vskip.5\baselineskip
		\(\begin{array}{ll}
			\nu X. G
			\\
			\equivby{\Cref{lem:absorb-link-creation}}
			G
			\\
			\equivby{E1}
			(\zero, G)
			\\
			\equivby{E4, E7}
			(\nu X. \nu X. X \bowtie X, G)
			\\
			\equivby{E4, \Cref{lem:absorb-link-creation}}
			(\nu X. X \bowtie X, G)             \\
			\hspace{2em}\because X \notin \mathit{fn}(\nu X. X \bowtie X)
			\\
			\equivby{E10}
			\nu X. (X \bowtie X, G)             \\
			\hspace{2em}\because X \notin \mathit{fn}(G)
			\\
			\equivby{The formar proof}
			\nu Y. (X \bowtie X, G)\angled{Y/X} \\
			\hspace{2em}\because X \in \mathit{fn}((X \bowtie X, G))
			\\
			=
			\nu Y. (Y \bowtie Y, G\angled{Y/X})
		\end{array}\)
		\\
		\(\begin{array}{ll}
			\equivby{E10}
			(\nu Y. Y \bowtie Y, G\angled{Y/X}) \\
			\hspace{2em}\because X \notin \mathit{fn}(G) \mbox{, thus } Y \notin \mathit{fn}(G\angled{Y/X})
			\\
			\equivby{E4, \Cref{lem:absorb-link-creation}}
			(\nu Y. \nu Y. Y \bowtie Y, G\angled{Y/X})
			\\
			\equivby{E7}
			(\zero, G\angled{Y/X})
			\\
			\equivby{E1}
			G\angled{Y/X}
			\\
			\equivby{\Cref{lem:absorb-link-creation}}
			\nu Y. G\angled{Y/X}
		\end{array}\)
	\end{quote}

\end{proof}



\begin{proof}[Proof of \Cref{th:symmetry-of-bowtie}]\mbox{}\\

	\begin{quote}
		\(\begin{array}{ll}
			\nu Z. (Z \bowtie X, Z \bowtie Y) &
			\\
			\equivby{E6}
			\nu Z. (X \bowtie Y)                \\
			\hspace{2em}\because (Z \bowtie Y)\angled{X/Z} = X \bowtie Y
			\\
			\equivby{\Cref{lem:absorb-link-creation}}
			X \bowtie Y
			\\
		\end{array}\)
	\end{quote}
	\vskip.5\baselineskip

	and

	\vskip.5\baselineskip
	\begin{quote}
		\(
		\begin{array}{ll}
			\nu Z. (Z \bowtie X, Z \bowtie Y) &
			\\
			\equivby{E2, E5}
			\nu Z. (Z \bowtie Y, Z \bowtie X) &
			\\
			\equivby{E6}
			\nu Z. (Y \bowtie X)                \\
			\hspace{2em}\because (Z \bowtie X)\angled{Y/Z} = Y \bowtie X
			\\
			\equivby{\Cref{lem:absorb-link-creation}}
			Y \bowtie X
			\\
		\end{array}
		\)
	\end{quote}
	\vskip.5\baselineskip

	Therefore,
	\(X \bowtie Y \equiv Y \bowtie X\).

\end{proof}

\section{Proof of properties of \(F_{GT}\)}

%
Theorem 4.1 (Soundness of \(F_{GT}\))
can be derived in the same way as in
the ordinary type systems for functional languages,
so we omit the precise proof.
Theorem 4.2 and Theorem 5.2 have a proof specific to \(F_{GT}\),
which is supplemented in this appendix.

\subsection{%
	Theorem 4.2 (\(F_{GT}\) and HyperLMNtal reduction)
}

\begin{lemma}\label{lem:redux-alpha}
	If
	\(G_1
	\rightsquigarrow^{\ast}_{P}
	G_2
	\)
	then
	\(G_1\angled{Y/X}
	\rightsquigarrow^{\ast}_{P}
	G_2\angled{Y/X}
	\)
\end{lemma}

\begin{proof}
	By (R1), (R2), (R3) and
	\(G_1
	\rightsquigarrow_{P}
	G_2
	\),
	we can show
	\(\nu X.(X \bowtie Y, G_1)
	\rightsquigarrow_{P}
	\nu X.(X \bowtie Y, G_2)
	\).
	Thus
	\(G_1\angled{Y/X}
	\rightsquigarrow_{P}
	G_2\angled{Y/X}
	\) by (R3).
	Then
	we can obtain
	\(G_1\angled{Y/X}
	\rightsquigarrow^{\ast}_{P}
	G_2\angled{Y/X}
	\)
	by induction on the length of the reduction
	\(\rightsquigarrow^{\ast}_{P}\)
\end{proof}

%
%
%
%
%

\begin{proof}[Proof of \Cref{th:fgt-stgamma}]

	We denote
	\({\overrightarrow{[\tau_i (\links{Y_i}) / x_i [\links{X_i}]]}}^{i}\)
	as
	\(\theta_x\)
	and
	\({[\overrightarrow{\tau_i (\links{Z_i}) / {(\lambda \dots)}_i (\links{W_i})}]}^{i}\)
	as
	\(\theta_\lambda\).

	We firstly prove \(\Rightarrow\).
	We split the cases by the last applied \(F_{GT}\) rules.

	\vspace{0.5em}\noindent%
	Case Ty-Ctx:
	\begin{quote}
		\(T = x[\links{X}]\)
		and
		\(x[\links{X}]: \tau\,(\links{X}) \in \Gamma\).
		Thus
		\(T[\tau\,(\links{X}) / x[\links{X}], \dots]\theta_\lambda
		= \tau\,(\links{X})
		\rightsquigarrow^{\ast}_{P}
		\tau\,(\links{X})
		\).
	\end{quote}

	\vspace{0.5em}\noindent%
	Case Ty-Arrow:
	\begin{quote}
		\(T = (\lambda \dots)(\links{X})\)
		where
		\(
		(\Gamma, P)
		\vdash
		(\lambda \dots)(\links{X}):
		\tau\,(\links{X})
		\).
		Thus
		\(T \theta_x [\tau\,(\links{X}) / (\lambda \dots)(\links{X})]
		= \tau\,(\links{X})
		\rightsquigarrow^{\ast}_{P}
		\tau\,(\links{X})
		\).
	\end{quote}

	\vspace{0.5em}\noindent%
	Case Ty-Cong:
	\begin{quote}
		Supppose the antecedent of Ty-Cong was
		\((\Gamma, P) \vdash T' : \tau\,(\links{X})\)
		where
		\(T \equiv T'\).
		By induction hypothesis,
		\(\tau\,(\links{X})
		\rightsquigarrow^{\ast}_{P}
		T' \theta_x \theta_\lambda
		\).
		Since
		\(
		T\theta_x \theta_\lambda
		\equiv
		T'\theta_x \theta_\lambda
		\),
		we can show
		\(\tau\,(\links{X})
		\rightsquigarrow^{\ast}_{P}
		T \theta_x \theta_\lambda
		\)
		using (R3).
	\end{quote}

	\vspace{0.5em}\noindent%
	Case Ty-Alpha:
	\begin{quote}
		Supppose the antecedent of Ty-Alpha was
		\((\Gamma, P) \vdash T' : \tau\,(\links{X'})\)
		where
		\(T = T'\angled{Y/X}\)
		and
		\(\tau\,(\links{X}) = \tau\,(\links{X'})\angled{Y/X}\).
		By induction hypothesis,
		\(\tau\,(\links{X'})
		\rightsquigarrow^{\ast}_{P}
		T' \theta_x \theta_\lambda
		\).
		Here,
		we can show that
		\(T\theta_x \theta_\lambda
		=
		T'\theta_x \theta_\lambda\angled{Y/X}
		\).
		Thefore, by \Cref{lem:redux-alpha},
		\(\tau\,(\links{X'})\angled{Y/X}
		\rightsquigarrow^{\ast}_{P}
		T \theta_x \theta_\lambda\angled{Y/X}
		\).
	\end{quote}

	\vspace{0.5em}\noindent%
	Case Ty-Prod:
	\begin{quote}
		Supppose the antecedents of Ty-Prod was
		\({\overrightarrow{(\Gamma, P) \vdash T_i : \tau_i (\links{X_i})}}^{i}\)
		where
		\(T = \Tauu[{\overrightarrow{T_i / \tau_i (\links{X_i})}}^{i}]\)
		By induction hypothesis,
		\(\tau_i (\links{X_i})
		\rightsquigarrow^{\ast}_{P}
		T_i \theta_x \theta_{\lambda i}
		\).
		Therefore, using (R1), (R2), and (R3),
		we can show
		\(\Tauu'
		\rightsquigarrow^{\ast}_{P}
		\Tauu'[T_i \theta_x \theta_{\lambda i} / \tau_i (\links{X_i})]
		\) for any \(\Tauu'\).
		Thus,
		we can have
		\(\tau_i (\links{X_i})
		\rightsquigarrow^{\ast}_{P}
		\Tauu_0
		\rightsquigarrow^{\ast}_{P}
		\dots
		\rightsquigarrow^{\ast}_{P}
		\Tauu_n
		\)
		where
		\(\Tauu_i\) is inductively defined as
		\(\Tauu_0 = \Tauu\)
		and
		\(\Tauu_{i + 1} = \Tauu_i[T_i \theta_x \theta_{\lambda i} / \tau_i (\links{X_i})]\),
		in which
		\(\Tauu_n = T\theta_x \theta_\lambda\).
	\end{quote}
	\vspace{1em}

	Then, we prove \(\Leftarrow\).
	by induction on the length of the reduction
	\(\rightsquigarrow^{\ast}_{P}\).
	We denote
	\({\overrightarrow{[x_i [\links{X_i}] / \tau_i (\links{Y_i})]}}^{i}\)
	as
	\(\theta_x^{-1}\)
	and
	\({[\overrightarrow{{(\lambda \dots)}_i (\links{W_i})} / \tau_i (\links{Z_i})]}^{i}\)
	as
	\(\theta_\lambda^{-1}\).
	Then, the proposition can be rewritten as
	\[
		\tau\, (\links{X})
		\rightsquigarrow^{\ast}_{P}
		\Tauu
		\Rightarrow
		(\Gamma, P)
		\vdash
		\Tauu \theta_x^{-1} \theta_\lambda^{-1}:
		\tau\,(\links{X}).
	\]

	\vspace{0.5em}\noindent%
	Case \(\tau\,(\links{X}) = \Tauu\)
	(The length of
	\(\rightsquigarrow^{\ast}_{P}\)
	is zero):
	\begin{quote}
		Follows by Ty-Ctx or Ty-Arrow
		depending on whether the
		\(\tau\,(\links{X})\)
		is replaced with the graph context in
		\(\theta_x^{-1}\)
		or the
		\(\lambda\)-abstraction atom in
		\(\theta_\lambda^{-1}\).
	\end{quote}

	\vspace{0.5em}\noindent%
	Case \(
	\tau\,(\links{X})
	\rightsquigarrow^{\ast}_{P}
	\Tauu'
	\rightsquigarrow_{P}
	\Tauu
	\)
	(The length of
	\(\rightsquigarrow^{\ast}_{P}\)
	is \(n > 0\)):
	\begin{quote}
		Suppose the production rule applied to reduce from
		\(\Tauu'\)
		to
		\(\Tauu\)
		was
		\(\alpha\,(\links{Y}) \longrightarrow \Tauu''\).
		Using (R1), (R2), and (R3),
		we can obtain (new)
		\(
		\tau\,(\links{X})
		\rightsquigarrow^{\ast}_{P}
		\Tauu'
		\rightsquigarrow_{P}
		\Tauu
		\)
		which satisfies
		\(
		\Tauu
		=
		\Tauu'[\Tauu'' / \alpha\,(\links{Y})]
		\).
		By induction hypothesis,
		we can obtain the derivation tree of
		\begin{align}\label{math:all}
			(\Gamma, P) \vdash \Tauu' \theta_x^{-1} \theta_\lambda^{-1} : \tau\,(\links{X}).
		\end{align}
		Since \(\Tauu'\) contains \(\alpha\,(\links{X})\),
		there exists a derivation of
		\begin{align}\label{math:prev}
			(\Gamma', P) \vdash \alpha \theta_x^{-1} \theta_\lambda^{-1} : \alpha\,(\links{X}).
		\end{align}
		in the tree.
		Since
		\begin{align}\label{math:after}
			(\Gamma', P) \vdash \Tauu'' \theta_x^{-1} \theta_\lambda^{-1} : \alpha\,(\links{X}).
		\end{align}
		holds immediately by Ty-Prod,
		we can replace the derivation tree of
		(\Ref{math:prev})
		with
		that of
		(\Ref{math:after})
		in
		that of
		(\Ref{math:all}),
		which will result in the derivation tree of the desired typing relation.

	\end{quote}

\end{proof}


\subsection{%
	Theorem 5.2 (decomposing graph with the last applied production rule)
}

We omit
\((\emptyset, P) \vdash\)
for brevity.

\begin{proof}[Proof of \Cref{th:graph-decomp}]
	We prove by induction on the derivation of \(G: \alpha (\links{Y})\)
	after the last application of Ty-Prod.

	By \Cref{lem:prod-needed},
	there exists the last Ty-Prod and only Ty-Cong and Ty-Alpha are used later
	on the derivation of \(G: \alpha (\links{Y})\).

	\vspace{0.5em}\noindent%
	Case Ty-Prod:
	\begin{quote}
		Trivial from the definition of Ty-Prod.
	\end{quote}

	\vspace{0.5em}\noindent%
	Case Ty-Cong:
	\begin{quote}
		The theorem holds on \(G: \alpha(\links{Y})\) by induction hypothesis.
		Therefore, it holds on \(G' \equiv G\).
	\end{quote}

	\vspace{0.5em}\noindent%
	Case Ty-Alpha:
	\begin{quote}
		By induction hypothesis,
		we can assume
		for \(G: \alpha (\links{Y})\),
		there exists
		\({\overrightarrow{G_j}}^{j}\)
		such that
		\(G \equiv\Tauu' [{\overrightarrow{G_j / \tau_j\, (\links{X_j})}}^{j}]\)
		where
		\begin{itemize}
			\item
			      \(\Tauu' = \Tauu {\overrightarrow{\angled{Y_i / X_i}}}^{i}\),

			\item
			      \(\tau_j\, (\links{X_j})\) are
			      all the type atoms
			      appearing in \(\Tauu'\),
			      and

			\item
			      \({\overrightarrow{G_j : \tau_j\, (\links{X_j})}}^{j}\).

		\end{itemize}

		For \(G\angled{Z/Y}: \alpha (\links{Y})\angled{Z/Y}\),
		we can obtain
		\begin{itemize}
			\item
			      \(\Tauu'' = \Tauu {\overrightarrow{\angled{Z_i / X_i}}}^{i}\),
			      where \(Z_i = Y_i\angled{Z/Y}\).
		\end{itemize}
		The type atom
		\(\tau_j\, (\links{Z_j})\)
		appearing in \(\Tauu''\),
		corresponding to the atom
		\(\tau_j\, (\links{X_j})\)
		in \(\Tauu'\),
		may have substituted its links.
		Thus, we need to denote it as
		\(\tau_j (\links{X_j})\theta_j\)
		where \(\theta_j\) is a hyperlink substitution
		which satisfies
		\(\tau_j (\links{X_j})\theta_j
		=\tau_j\, (\links{Z_j})\).
		Since
		\({\overrightarrow{G_j : \tau_j\, (\links{X_j})}}^{j}\)
		holds by the induction hypothesis,
		we can show that
		\({\overrightarrow{G_j\theta_j : \tau_j\, (\links{X_j}) \theta_j}}^{j}\)
		holds using Ty-Alpha.
		Therefore, we can obtain
		\({\overrightarrow{G_j\theta_j}}^{j}\)
		that satisfies the conditions.
	\end{quote}

\end{proof}

\begin{biography}

	\profile{Jin Sano}{%
		received his B.Eng.\ degree from Waseda University in 2021.
		His research interests include design and implementation of
		programming languages and type systems,
		and software verifications.
	}
	\profile{Naoki Yamamoto}{%
		received his B.Eng.
		and M.Eng.
		degrees from Waseda University in 2019 and 2021, respectively.
		He has been in a doctoral course at Waseda University since 2021.
		He has been a member of Waseda Chapter (Mu-Tau) of IEEE-HKN (Eta-Kappa-Nu) since 2019.
		His research interests include programming languages and program verification by proof assistants.
	}

	\profile{Kazunori Ueda}{%
		received his M.Eng.
		and Dr.Eng.
		degrees from the University of Tokyo in 1980 and 1986, respectively.
		He joined NEC in 1983, and from 1985 to 1992,
		he was with the Institute for New Generation Computer  Technology (ICOT) on loan.
		He joined Waseda University in 1993 and has been Professor since 1997.
		He is also Visiting Professor of Egypt-Japan University of Science and Technology since 2010.
		His research interests include design and implementation of programming languages,
		concurrency and parallelism, high-performance verification, and hybrid systems.
	}

\end{biography}

\end{document}